\documentclass[journal]{IEEEtranTCOM}
\usepackage[dvipsnames]{xcolor}
\normalsize

\usepackage{mathtools}
\usepackage[linesnumbered,lined,ruled,commentsnumbered,onelanguage,algo2e]{algorithm2e}
\usepackage{amsmath}
\usepackage{amssymb}
\usepackage{amsthm}
\usepackage{tabularx}
\usepackage{multirow}
\usepackage{graphicx}
\usepackage{subcaption} 
\usepackage[utf8]{inputenc}
\usepackage[export]{adjustbox}
\usepackage{wrapfig}
\usepackage{setspace}
\usepackage[Symbol]{upgreek}
\usepackage{setspace,lipsum}
\usepackage{url}
\usepackage[T1]{fontenc} 
\usepackage{cite}
\usepackage{float}
\usepackage{algpseudocode}

\usepackage{lipsum}

\newcommand{\mF}{\mathcal{F}}
\newcommand{\mA}{\mathcal{A}}

\newcommand{\mR}{\mathcal{R}}

\DeclareMathOperator{\wt}{wt}

\DeclareMathOperator*{\argmax}{arg\,max} 
\DeclareMathOperator*{\argmin}{arg\,min} 

\DeclareMathAlphabet{\mathcalbf}{OMS}{pzc}{b}{n}

\newtheorem{lemma}{Lemma}
      \newtheorem{theorem}{Theorem}

\newtheorem{remark}{Remark}

\newtheorem{definition}{Definition}
\newtheorem{example}{Example}

\begin{document}
\newcommand{\set}[1]{\left\{{#1}\right\}}
\title{Frozen Set Design for Precoded Polar Codes}
\author{Vera~Miloslavskaya,~Yonghui~Li,~\IEEEmembership{Fellow,~IEEE},~and~Branka~Vucetic,~\IEEEmembership{Life~Fellow,~IEEE}
\thanks{Vera Miloslavskaya is with the School of Science and Technology, The University of New England, Australia (e-mail: Vera.Miloslavskaya@une.edu.au) and the School of Electrical and Computer Engineering, The University of Sydney, Australia (e-mail: vera.miloslavskaya@sydney.edu.au)}  
\thanks{Yonghui Li and Branka Vucetic are with the School of Electrical and Computer Engineering, The University of Sydney, Australia (e-mail:
yonghui.li@sydney.edu.au, branka.vucetic@sydney.edu.au).}
\thanks{This research was supported by the Australian Research Council under Grants FL160100032,
DP190101988 and DP210103410.}}


\maketitle


\begin{abstract}

This paper focuses on the frozen set design for precoded polar codes decoded by the successive cancellation list (SCL) algorithm. We propose a novel frozen set design method, whose computational complexity is low due to the use of analytical bounds and constrained frozen set structure. We derive new bounds based on the recently published complexity analysis of SCL {decoding} with near maximum-likelihood (ML) performance. To predict the ML performance, we employ the state-of-the-art bounds relying on the code weight distribution. The bounds and constrained frozen set structure are incorporated into the genetic algorithm to generate optimized frozen sets with low complexity. Our simulation results show that the constructed precoded polar codes of length $512$ have a superior frame error rate (FER) performance compared to the state-of-the-art codes under SCL decoding with various list sizes.

\end{abstract}
\IEEEpeerreviewmaketitle 
\IEEEpeerreviewmaketitle
\section{Introduction}

The polar codes 
\cite{arikan2009channel} have frozen bits that are all {set to fixed values. Their} generalizations such as the CRC-aided polar codes \cite{polarNR2018}, polar subcodes \cite{trifonov2016subcodes}, parity-check-concatenated polar codes \cite{wang2016parcheckconcpolar}, polarization-adjusted convolutional (PAC) codes \cite{Arkan2019FromSD} and precoded polar codes \cite{milos2020precoded} involve frozen bits with non-fixed values, whose computation may be specified by linear combinations of information bits with lower indices. These combinations are referred to as the frozen bit expressions. 
Since polar codes with near-uniformly distributed frozen bit expressions are known to perform well \cite{trifonov2017randsubcodes,Coskun2022InfTheor}, 
we limit our consideration to such codes. Their design problem reduces to the frozen set design problem. 

We treat the frozen set design problem as an optimization problem with the objectives of minimizing the decoding error probability and 
complexity.  
For any particular decoder, the frozen set may be optimized by using the genetic algorithm \cite{Elkelesh2019}, where the code performance is evaluated via decoding simulations. However, the inherent high computational complexity of these simulations necessitates a shift towards analytical methods for code evaluation to ensure computational efficiency. 
The state-of-the-art analytical methods for the polar code evaluation are as follows. The frame error rate (FER) of polar codes under the successive cancellation (SC) decoding \cite{arikan2009channel} can be predicted using \cite[Eq. (3)]{PeSC2014}. For the maximum-likelihood (ML) decoding, there are the FER bounds \cite{Sason2006PeMLtutorial} parameterized by the weight distribution that can be computed using \cite{Canteaut1998ANA,PreTransformedSpectrum2021,milos2022, Yao2023determ}. Although there is no analytical bound predicting the FER under the SC list (SCL) decoder \cite{tal2015list}, the average list size required by SCL {decoding} to approach the ML performance can be characterized by the information-theoretical quantities \cite{Coskun2022InfTheor}. The ML performance may also be approached by the Fano decoding \cite{Fano1963}, whose complexity 
is {related to} the cutoff rate \cite{Arıkan2016OriginOfPolar}. 
We focus on the SCL decoder as the most widely used decoder for precoded polar codes.   


In this paper, 
we propose a novel low-complexity frozen set design method for precoded polar codes with various tradeoffs between the FER performance and decoding complexity.
The main contributions are as follows. First, we explore 
the SCL list size lower bound from \cite{Coskun2022InfTheor} and 
identify the factors limiting its effectiveness as the 
predictive measure for
near ML decoding complexity. Second, we improve the 
prediction accuracy by tightening the lower bound from \cite{Coskun2022InfTheor}. Third, we introduce an approximate lower bound that facilitates a fair comparison of various frozen sets. This approximation 
combines our tightened lower bound with the upper bound from \cite{Coskun2022InfTheor}. 
Fourth, we propose to solve the frozen set optimization problem 
by minimizing the ML decoding error probability estimate under the decoding complexity constraint, which is given by the proposed approximate lower bound. The resulting frozen sets are intended for precoded polar codes utilizing frozen bit expressions with near-uniformly distributed binary coefficients.
Fifth, we impose constraints on the frozen set structure to reduce 
the optimization complexity.
Our simulation results show that the constructed precoded polar codes of length $512$ have a superior FER performance compared to the state-of-the-art codes under SCL decoding with various list sizes. This confirms the efficiency of the proposed approximate lower bound 
as the ML decoding complexity measure for comparing various frozen sets. 
Given an approximate lower bound value, the frozen set optimization complexity 
is low due to the constraints on the frozen set structure and {no need to perform} 
decoding simulations. For example, the genetic algorithm requires less than a minute to solve this problem for the code length $512$. 

The paper is organized as follows. Section \ref{sPreliminaries} provides a background on the polar codes and relevant frozen set design criteria. In Section \ref{sProposed}, we derive the proposed bounds and specify the corresponding frozen set optimization process. In Section \ref{sNumerical}, we present the numerical results on the frozen set design complexity 
and the FER performance of precoded polar codes with the proposed frozen sets and compare them with the state-of-the-art.

\section{Preliminaries}
\label{sPreliminaries}

This section provides a background on the polar codes, the ML performance of precoded polar codes, and the complexity of near ML decoding using {the SCL decoder}.

\subsection{Polar Codes}

An $(N=2^n,K)$ polar code \cite{arikan2009channel} is a binary linear block code consisting of codewords\footnote{We omit the multiplication by the bit-reversal permutation matrix $B$ since $u\cdot B\cdot G^{\otimes n}=u\cdot G^{\otimes n}\cdot B$ and the proposed techniques can be easily applied to permuted polar codes as well.} $c=u\cdot G^{\otimes n}$, where   $G=\big(\begin{smallmatrix}1&0\\1&1\end{smallmatrix}\big)$, $\otimes n$ denotes the $n$-fold Kronecker product, the input vector $u$ 
has $K$ information bits $u_i$, $i\in\mA$, and $N-K$ frozen bits $u_i$, $i\in\mF=[N]\setminus\mA$, and $[N]\triangleq \{ 0,\dots, N-1\}$. Note that $[N]=\emptyset$ for $N\leq 0$. The sets $\mA$ and $\mF$ are referred to as the information and frozen sets, respectively. In the case of the original polar codes \cite{arikan2009channel}, all frozen bits have fixed values, e.g., zeros.

In a more general case, the frozen bits are equal to linear combinations of the other input bits with lower indices \cite{trifonov2013polar}, known as the \textit{\textbf{frozen bit expressions}}. The resulting polar codes are referred to as the polar codes with dynamic frozen bits, parity-check concatenated polar codes, precoded polar codes and pre-transformed polar codes in the literature. We use the term ``precoded polar codes'' as in our previous works \cite{milos2020precoded,milos2021recursive}.

\subsection{Weight Distribution of Precoded Polar Codes and Their ML Performance}
\label{sPeML}

The precoded polar codes are linear codes and, therefore, {their performance under ML decoding depends on} their weight distributions.
However, the complexity of computing the exact weight distribution is high, except for very short codes and well-structured codes. In this paper, we employ the {ensemble-averaged} weight distribution \cite{PreTransformedSpectrum2021}. 
 {Specifically, the weight distribution is averaged over the ensemble of precoded polar codes with a given frozen set and all possible frozen bit expressions. 
Given the code parameters $(N,K)$ and information set $\mA$, the average number $\overline{W}_{t}$ of codewords with weight $t$ is computed as
\begin{align*}
    &\overline{W}_{t}=\sum_{\substack{i\in\mA \\ 2^{\wt(i)}\leq t}} 2^{K-|\mA\cap \{ 0,\dots,i\}|} P(N,i,t), \label{eq:W} \\
 &P(N,i,t)=\nonumber\\
 &\quad\quad\begin{cases}
    \displaystyle\sum_{\substack{t'=2^{\wt(i)} \\ t-t'\mathrm{ is\, even}}}^{\min(t, N/2)} P(N/2,i,t')\frac{\binom{N/2-t'}{(t-t')/2}}{2^{N/2-t'}},\quad i\in [N/2],
    \nonumber \\ 
    P(N/2,i-N/2,t/2),\quad i\in [N]\setminus [N/2], t \mathrm{\,is\, even}, \\ 0,\quad i\in [N]\setminus [N/2], t \mathrm{\,is\, odd},\end{cases}
\end{align*}
where $\wt(i)$ is the Hamming weight of the binary expansion of integer $i$.
The boundary conditions: $P(2,0,1)=P(2,1,2)=1$, and $P(N,i,t)=0$ when ($i=0$ and $t$ is even) or ($i>0$ and $t$ is odd). We recommend a log-domain implementation for accuracy. 
\begin{table*}
\centering{
\caption{{The weight distribution of randomly precoded polar codes of length $128$ 
} 
}
\label{tab:randPrecodedPolar}
{
\centering{
\scriptsize
\begin{tabular}{|r|r|r|r||r|r|r|r||r|r|r|r|} \hline
\multicolumn{4}{|c||}{\multirow{2}{*}{$(128,48)$ frozen set}}& \multicolumn{4}{c||}{\multirow{2}{*}{$(128,64)$ frozen set}}&    \multicolumn{4}{c|}{\multirow{2}{*}{$(128,80)$ frozen set}}  \\
\multicolumn{4}{|c||}{}& \multicolumn{4}{c||}{}&  \multicolumn{4}{c|}{}
\\ \hline
\multirow{2}{*}{$t$} &     Ensemble- & Code (a) &   Code (b) &
\multirow{2}{*}{$t$} &     Ensemble- & Code (a) &   Code (b) &
\multirow{2}{*}{$t$} &     Ensemble- & Code (a) &   Code (b) \\
& averaged $\overline{W}_{t}$ & exact $W_{t}$ & exact $W_{t}$ &
& averaged $\overline{W}_{t}$ & exact $W_{t}$ & exact $W_{t}$ &
& averaged $\overline{W}_{t}$ & exact $W_{t}$ & exact $W_{t}$ \\ \hline
16&1864&1848&1880& 8&272&264&248& 8&4308&4328&4320   \\ \hline
20&17050&17184&16992& 12&896&928&992& 10&2016&2624&2304   \\ \hline
22&405&384&128& 16&85423&74984&79688& 12&363408&372576&363584 \\ \hline
24&306960&307136&310720& 18&6104& 5760&5760& 14&1077792&1212352&1120512\\ \hline
26&40132&40320&25728  & & & & & & & & \\ \hline
28&3399934&3409888&3501728& & & & & & & & \\ \hline 
30&1725681&1725696&1349376& & & & & & & & \\ \hline 
\end{tabular} 
}}}
\end{table*}

The complexity of computing $P(N,i,t)$ scales as $O(N)$, assuming that all $P(N/2,\cdot,\cdot)$ are available. This leads to the complexity $O(N^3)$ for computing $P(N,i,t)$ over all possible $i$ and $t$. Thus, the worst-case complexity of computing the ensemble-averaged weight distribution scales as $O(N^3)$, as explained in \cite[Section III-C]{PreTransformedSpectrum2021}. Table \ref{tab:randPrecodedPolar} compares the exact partial weight distribution of several precoded polar codes \cite{milos2022}\footnote{{We consider three frozen sets that are given by the less reliable bit-channels according to the  Gaussian approximation \cite{trifonov2012efficient} for AWGN, BPSK, $E_b/N_0=4$. For each frozen set, we produce two precoded polar codes by randomly generating frozen bit expressions. 
}} with the ensemble-averaged weight distribution \cite{PreTransformedSpectrum2021} to illustrate the accuracy of the latter one. 
} 
To estimate the ML decoding error probability, we substitute the {ensemble-}averaged weight distribution into the union bound \cite{Sason2006PeMLtutorial}, known for its simplicity, and the tangential-sphere bound (TSB) \cite{poltyrev1994bounds}, known for its tightness. 

\subsection{Complexity of SCL Decoding with Near ML Performance}
\label{sOrigDm}


It has been shown in \cite{tal2015list} that the time complexity of SCL is $O(L N \log(N))$ and its space complexity is $O(L N)$, where $L$ is the decoding list size. 
The FER performance of SCL decoding was experimentally shown to improve with increasing $L$ at the expense of increasing complexity.
Recently, \cite{Coskun2022InfTheor} provided ground-breaking results on the list size $L$ such that SCL has a near ML performance.
For general binary memoryless symmetric (BMS) channels, \cite[Theorem 1]{Coskun2022InfTheor} 
proved that the mean value of the binary logarithm of $L$ required at the $m$-th stage of SCL to achieve the ML performance is upper bounded by the conditional entropy $\bar D_m$
\begin{equation}
    \bar D_m \triangleq H(U_{\mA^{(m)}}|Y_{[N]}, U_{\mF^{(m)}}),
    \label{eq:barDm}
\end{equation}
where $m\in [N]$, $\mA^{(m)}\triangleq \{i\in \mA \,|\, i\leq m\}$, $\mF^{(m)}\triangleq \{i\in \mF \,|\, i\leq m\}$, $U_{T}\triangleq \{ U_i \, | \, i\in T\}$ for any set $T$, $U_i$ is the random variable corresponding to the $i$-th input bit, and $Y_i$ is the random variable corresponding to the $i$-th output. 
Note that we use the notation of \cite{Coskun2022InfTheor} except for starting enumeration from zero instead of one.
Unfortunately, the computation of $\bar D_m$ requires performing decoding with a huge/unbounded list size as pointed out in \cite[Remark 2]{Coskun2022InfTheor}. To overcome this issue, \cite[Remark 2]{Coskun2022InfTheor} suggested to characterize the decoding list size using the lower bound on $\bar D_m$ that is derived in \cite[Section III-A]{Coskun2022InfTheor}. 
This lower bound is defined as
$
\bar D_m \geq\sum_{i\in\mA^{(m)}} H_{n,i} - \sum_{i\in\mF^{(m)}} (1-H_{n,i})$ by \cite[Eq. (6a)]{Coskun2022InfTheor}, where $H_{n,i}$ is the entropy of the $i$-th bit-channel, $i\in [2^n]$. 
However, it follows from the numerical results \cite[Fig. 1]{Coskun2022InfTheor} that the actual lower bound on $\bar D_m$, denoted by us as $\bar D^{\mathrm{low}}_{m}$, takes into account the non-negativity of entropy in Eq. \eqref{eq:barDm} as 
\begin{equation}
\bar D^{\mathrm{low}}_{m}=\begin{cases} \bar D^{\mathrm{low}}_{m-1}+H_{n,m}, & m\in\mA, \\ \max(\bar D^{\mathrm{low}}_{m-1}-(1-H_{n,m}), 0), & m\in\mF,\end{cases}
    \label{eq:lower}
\end{equation}
where $m\in [2^n]$, and $\bar D^{\mathrm{low}}_{-1}=0$.
Note that $H_{n,i}$ can be represented as $1-I_{n,i}$, where $I_{n,i}$ is the mutual information of the $i$-th bit-channel that can be recursively {estimated} using {the analytical approximation} \cite[Eqs. (9), (10) and (26)]{Brannstrom2005J} for the AWGN channel with BPSK modulation. {The application of this approximation to polar codes can be found in \cite[Eqs. (4.1), (4.2), (4.6) and (4.7)]{dosio2016MutInf}.} 

\subsection{Frozen Bit Expressions}
\label{sFrBitExpressions}
It has been shown that codes with randomly generated frozen bit expressions can perform well \cite{trifonov2017randsubcodes,trifonov2020randsubcodes,Coskun2022InfTheor,PreTransformedSpectrum2021}. However, the random generation limits the reproducibility of the results. Following \cite{milos2024Adapt}, we ensure the reproducibility by using the deterministic binary sequence $\omega$ produced from the rational approximation of the $\pi$ number: $\pi \approx\frac{104348}{33215}$. Thus, $\omega=(\omega_0,\omega_1,\omega_2,\omega_3,\dots)$ is equal to the binary expansion of $\frac{104348}{33215}$ that can be easily computed. 
Given $\omega$ and the information bits $u_i$, $i\in\mA$, we calculate the values of the frozen bits $u_i$, $i\in\mF$, as follows:
\begin{algorithmic}
\small
\setstretch{0.8}
\State $b\leftarrow 0$
\State \textbf{for} {$i\in\mF$} \textbf{do} 
    \State \quad\quad\quad {$u_i\leftarrow 0$}
    \State \quad\quad\quad \textbf{for} $j\in\mA$, $j<i$ \textbf{do} 
        \State \quad\quad\quad \quad\quad  $u_i\leftarrow u_i+\omega_b\cdot u_{j}$
        \State \quad\quad\quad \quad\quad $b \leftarrow b+1$
\end{algorithmic}

\section{Proposed Frozen Set Design for Precoded Polar Codes}
\label{sProposed}

This section presents our low-complexity frozen set design method for precoded polar codes with various tradeoffs between the FER performance and decoding complexity. 
We focus on the problem of the complexity prediction for SCL with a near ML performance,  
since this problem has been partially solved by 
$\bar D^{\mathrm{low}}_{m}$ from Eq. \eqref{eq:lower}. 

This section is organized as follows. We first consider limitations of $\bar D^{\mathrm{low}}_{m}$ as a decoding complexity measure in Section \ref{sCoskunLim} and identify their source in Section \ref{sCoskunDeriv}. 
To resolve the identified issues, we derive a new tightened lower bound $\bar D^{\mathrm{tight}}_{m}$ in Section \ref{sTightened} and alleviate the influence of the frozen set structure by combining the tightened lower bound with an upper bound in Section \ref{sApproximate}.    
The resulting approximate bound 
$\bar D^{\mathrm{apx}}_{m}$ is further used as a decoding complexity measure during the 
frozen set optimization in Section \ref{sOptimization}. The optimization complexity is significantly reduced by imposing constraints on the frozen set structure.
Note that the proposed frozen design approach is intended for precoded polar codes with near-uniform frozen bit expressions since both the performance and complexity criteria have been derived for such codes. 



\subsection{Limitations of $\bar D^{\mathrm{low}}_{m}$ as a Decoding Complexity Measure}
\label{sCoskunLim}

The necessity to have a low $\bar D^{\mathrm{low}}_{m}$ for a precoded polar code to approach the ML performance under SCL with a low complexity has been proven in \cite{Coskun2022InfTheor}  for BMS channels. 
Besides, \cite[Appendix]{Coskun2022InfTheor} specified three exemplary frozen sets for the code parameters $(512,256)$ and \cite[Fig. 4]{Coskun2022InfTheor} illustrated their remarkable performance. 
However, the following example shows the limited applicability of $\bar D^{\mathrm{low}}_{m}$ for the frozen set comparison\footnote{Besides, there are two inherent weaknesses of the entropy-based analysis of the SCL decoder that are described in \cite[Remark 1]{Coskun2022InfTheor}.}.  
For the code parameters $(512,256)$, the frozen set consisting of $256$ less reliable bit-channels is characterized by $\max_{m} \bar D^{\mathrm{low}}_{m}=0.953$, where the bit-channel reliabilities are calculated by the Gaussian approximation \cite{trifonov2012efficient} for AWGN, BPSK, and $E_b/N_0=2$ dB. {According to our experimental results for the corresponding precoded polar code under SCL decoding with list size $L$ at $E_b/N_0=0.5$ dB, the mean value of the binary logarithm of $L$ required to achieve the ML performance is $1.7$. The same $\max_{m} \bar D^{\mathrm{low}}_{m}=0.953$ is provided by another $(512,256)$ frozen set for which the experimentally obtained mean value of the binary logarithm of $L$ is about $4.5$.} 
The existence of $(N,K)$ precoded polar codes with similar $\max_m \bar D^{\mathrm{low}}_{m}$ but different complexities of near ML decoding hinders the usage of $\bar D^{\mathrm{low}}_{m}$ as the decoding complexity measure during the frozen set optimization for SCL.

\subsection{Derivation of $\bar D^{\mathrm{low}}_{m}$ in \cite{Coskun2022InfTheor}}
\label{sCoskunDeriv}

The source of the issues with $\bar D^{\mathrm{low}}_{m}$  follows from its derivation in \cite[Section III-A]{Coskun2022InfTheor}. Specifically, the lower bound $\bar D^{\mathrm{low}}_{m}$ on $\bar D_m$ is obtained for BMS channels by introducing $\Delta_m\triangleq \bar D_m-\bar D_{m-1}$ and showing that $\Delta_m=H(U_m|Y_{[N]}, U_{[m]})$ when $m\in\mA$ and $\Delta_m=H(U_m|Y_{[N]}, U_{[m]})-H(U_m|Y_{[N]}, U_{\mF^{(m-1)}})\geq H(U_m|Y_{[N]}, U_{[m]})-1$ when $m\in\mF$. Thus, the gap between $\bar D_m$ and its lower bound $\bar D^{\mathrm{low}}_{m}$ is due to 
replacing $H(U_m|Y_{[N]}, U_{\mF^{(m-1)}})$ by its upper bound 1 when $m\in\mF$. At the same time, $H(U_m|Y_{[N]}, U_{\mF^{(m-1)}})$ is lower bounded by $H(U_m|Y_{[N]}, U_{[m]})$, which means that $\Delta_m\leq 0$ when $m\in\mF$ and leads to the upper bound $\bar D_m\leq \sum_{{i}\in\mA^{(m)}} H_{n,i}$ \cite[Eq. (6b)]{Coskun2022InfTheor}. \cite[Remark 2]{Coskun2022InfTheor} explains the preferability of the lower bound on $\bar D_m$ compared to the upper bound by the fact that the upper bound ignores the effect of the frozen bits. 

\subsection{Proposed Tightened Lower Bound $\bar D^{\mathrm{tight}}_{m}$} 
\label{sTightened}

We propose to tighten the lower bound on $\bar D_m$ by tightening the upper bound on $H(U_m|Y_{[N]}, U_{\mF^{(m-1)}})$. Observe that $H(U_m|Y_{[N]}, U_{\mF^{(m-1)}})$ is upper bounded by $H(U_m|Y_{T}, U_{\Phi})$ for any subsets $\Phi\subseteq \mF^{(m-1)}$ and $T\subseteq [N]$. In what follows below we show how to identify non-trivial sets $\Phi$ and $T$ such that $H(U_m|Y_{T}, U_{\Phi})$ can be easily computed. The following example illustrates the case of $N=4$. 

\begin{example}
\label{exam:N4}
    For $n=2$ and $N=2^n=4$, the $N\times N$ polar transformation\footnote{
The bit reversal permutation matrix $B$ can be easily incorporated by permuting elements of $Y_{[N]}$, i.e., by replacing $Y_{[N]}$ with $Y_{[N]} B$. 
} is specified by $G^{\otimes n}=\left(\begin{smallmatrix} 1&0&0&0 \\ 1&1&0&0 \\ 1&0&1&0 \\ 1&1&1&1 \end{smallmatrix}\right)$. 
Let us consider various cases of $\mF^{(m-1)}$ and calculate the corresponding upper bounds on $h_{m,\mF}\triangleq H(U_m|Y_{[N]},U_{\mF^{(m-1)}})$ 
\newline
    Case $m=0\colon$ 
    \begin{itemize}
    \item $\mF^{(m-1)}=\emptyset$ and then $h_{0,\mF}=H(U_0|Y_{[N]})=H_{n,0}$ by the definition of $H_{n,m}$.
    \end{itemize}
    Case $m=1\colon$
        \begin{itemize}
        \item if $\mF^{(m-1)}=\{ 0 \}$, then $h_{1,\mF}=H(U_1|Y_{[N]},U_0)=H_{n,1}$ by the definition of $H_{n,m}$. 
        \item if $\mF^{(m-1)}=\emptyset$, then $h_{1,\mF}=H(U_1|Y_{[N]})\leq 
        H(U_1|Y_1,Y_3)=H_{n-1,0}$ since the received vector $(Y_1,Y_3)$ corresponds to the transmitted $(U_1,U_3)G$.
        \end{itemize}
    Case $m=2\colon$
        \begin{itemize}
        \item if $\mF^{(m-1)}=\{ 0,1 \}$, then $h_{2,\mF}=H(U_2|Y_{[N]},U_0,U_1)=H_{n,2}$ by the definition of $H_{n,m}$.
        \item if $\mF^{(m-1)}\in\{\emptyset, \{ 0\}, \{ 1\} \}$, then $h_{2,\mF}=H(U_2|Y_{[N]}, U_{\mF^{(m-1)}})\leq H(U_2|Y_2,Y_3)=H_{n-1,0}$ since the received vector $(Y_2,Y_3)$ corresponds to the transmitted $(U_2,U_3)G$.
        \end{itemize}
    Case $m=3\colon$
        \begin{itemize}
        \item if $\mF^{(m-1)}=\{ 0,1,2 \}$, then $h_{3,\mF}=H(U_3|Y_{[N]},U_0,U_1,U_2)=H_{n,3}$ by the definition of $H_{n,m}$.
        \item if $\mF^{(m-1)}\in\{\emptyset, \{ 0\}\}$, then $h_{3,\mF}=H(U_3|Y_{[N]},U_{\mF^{(m-1)}})\leq H(U_3|Y_3)=H_{n-2,0}$ since the received $Y_3$ corresponds to the transmitted $U_3$.
        \item if $\mF^{(m-1)}\in\{ \{ 2\}, \{ 0,2\}, \{ 1,2\} \}$, then $h_{3,\mF}=H(U_3|Y_{[N]}, U_{\mF^{(m-1)}})\leq  H(U_3|Y_2,Y_3,U_2)=H_{n-1,1}$ since the received vector $(Y_2,Y_3)$ corresponds to the transmitted $(U_2,U_3)G$.
        \item if $\mF^{(m-1)}\in\{ \{ 1\}, \{ 0,1\} \}$, then $h_{3,\mF}=H(U_3|Y_{[N]}, U_{\mF^{(m-1)}})\leq H(U_3|Y_1,Y_3, U_1)=H_{n-1,1}$ since the received vector $(Y_1,Y_3)$ corresponds to the transmitted $(U_1,U_3)G$. 
        \end{itemize}

\end{example}

Example \ref{exam:N4} specifies the upper bounds on $H(U_m|Y_{[N]},U_{\mF^{(m-1)}})$ for $N=4$. The following lemma defines 
the upper bound on $H(U_m|Y_{[N]},U_{\mF^{(m-1)}})$ for 
any given $N=2^n$, $m$ and $\mF^{(m-1)}$. 
Let $T_{I,J}$ be a submatrix of $T$ consisting of the elements $T_{i,j}$, $i\in I$, $j\in J$. 

\begin{lemma}
\label{lem:H}
Let sets $I,J\subseteq [2^n]$ and integer $\widetilde n \leq n$ satisfy the following conditions:
\begin{enumerate}
\item $|I|=|J|=2^{\widetilde n}$, 
\item $(G^{\otimes n})_{I,J}=G^{\otimes \widetilde n}$, \item $(G^{\otimes n})_{\overline I,J}=\mathbf{0}$, 
\item $m\in I$, 
\item $I\cap [m] \subseteq \mF^{(m-1)}$. 
\end{enumerate}
Then 
 \begin{equation}
    H(U_m|Y_{[2^n]},U_{\mF^{(m-1)}}) \leq H_{\widetilde n, \widetilde  m},
    \label{eq:lemH}
\end{equation}
where $\widetilde m \triangleq | I\cap [m] |$, $\overline I \triangleq [2^n]\setminus I$, and $\mathbf{0}$ is all-zero matrix/vector.
\end{lemma}

\begin{proof}
For any such $I$ and $J$, we have $
H(U_m|Y_{[2^n]},U_{\mF^{(m-1)}})\leq H(U_m|Y_J,U_{I\cap [m]})$ due to $J\subseteq [2^n]$ and condition 5: $I\cap [m]\subseteq \mF^{(m-1)}$.  
By substituting the random variable vectors ${\widetilde U}_{[2^{\widetilde n}]}\triangleq U_{I}$ and ${\widetilde Y}_{[2^{\widetilde n}]}\triangleq Y_J$, we obtain $H(U_m|Y_J,U_{I\cap [m]})\overset{(a)}{=} H({\widetilde U}_{\widetilde m}|{\widetilde Y}_{[2^{\widetilde n}]}, {\widetilde U}_{[\widetilde m]}) \overset{(b)}{=} H_{\widetilde n, \widetilde m}$. Equality $(a)$ holds since ${\widetilde U}_{\widetilde m}=U_m$ and ${\widetilde U}_{[\widetilde m]}=U_{I\cap [m]}$ due to the definition of $\widetilde m$ 
and condition 4: $m\in I$. Equality $(b)$ holds since the received $Y_J={\widetilde Y}_{[2^{\widetilde n}]}$ corresponds to the transmitted $U (G^{\otimes n})_{[2^n],J}= \underbrace{U_I}_{{\widetilde U}_{[2^{\widetilde n}]}} \underbrace{(G^{\otimes n})_{I,J}}_{ G^{\otimes \widetilde n}}\oplus U_{\overline I} \underbrace{(G^{\otimes n})_{\overline I,J}}_{ \mathbf{0}}= {\widetilde U}_{[2^{\widetilde n}]}G^{\otimes \widetilde n}$ due to conditions 1--3. This concludes the proof. 
\end{proof}

The upper bound of Lemma \ref{lem:H} is non-constructive since it does not specify how to find the sets $I$ and $J$. The following two lemmas define sets $I$ and $J$ satisfying conditions 2--3 of Lemma \ref{lem:H}:  $(G^{\otimes n})_{I,J}=G^{\otimes \widetilde n}$ and $(G^{\otimes n})_{\overline I,J}=\mathbf{0}$. Lemma \ref{lem:setsIJ} considers the case of $|I|=|J|=2^{n-1}$, and then Lemma \ref{lem:recSetsIJ} generalizes the result for $|I|=|J|=2^{\widetilde n}$, $\widetilde n\leq n$. Note that we employ the binary representation $(j_0,\dots,j_{n-1})\in\{0,1\}^n$ of the integers $j=\sum_{t=0}^{n-1} j_t 2^t\in [2^n]$.

\begin{lemma}
\label{lem:setsIJ}
    Given any integer $q\in [n]$ and the corresponding set 
    $$S(q)\triangleq \Big\{ 
    j\in [2^n] \;|\; j_q=1\Big\},$$
    where $j_q$ is the $q$-th bit in the binary expansion of the integer $j$. Then sets $I=J=S(q)$ satisfy the conditions $(G^{\otimes n})_{I,J}=G^{\otimes (n-1)}$ and $(G^{\otimes n})_{\overline I,J}=\mathbf{0}$.
\end{lemma}

\begin{proof}
As shown in \cite{Bardet2016}, the $j= \sum_{t=0}^{n-1} j_t 2^t$-th row of $\left(\begin{smallmatrix} 1&1\\ 0&1 \end{smallmatrix}\right)^{\otimes n}$ can be represented as the evaluation of polynomial $f_n(j,x)\triangleq x_0^{j_0}x_1^{j_1}\cdots x_{n-1}^{j_{n-1}}$ over $2^n$ elements $x\triangleq \sum_{t=0}^{n-1} x_t 2^t\in [2^n]$. Since the $j$-th column of $G^{\otimes n}$ is equal to the transposed $j$-th row of $\left(\begin{smallmatrix} 1&1\\ 0&1 \end{smallmatrix}\right)^{\otimes n}$, it has the same polynomial representation. Thus, columns of $G^{\otimes n}$ with the indices $j\in J=S(q)$ correspond to polynomials $x_0^{j_0}\cdots 
x_q^{j_q=1} 
\cdots x_{n-1}^{j_{n-1}}$. For all $x\in \overline I=[2^n]\setminus S(q)$, the multiplier $x_q=0$ due to the definition of $S(q)$. Consequently, we have $f_n(j,x)=0$ for all $x\in \overline I$, $j\in J$. 
Therefore, the condition $(G^{\otimes n})_{\overline I,J}=\mathbf{0}$ is satisfied. For all $x\in I=S(q)$, the multiplier $x_q=1$ and consequently  $f_n(j,x)/x_q=x_0^{j_0}\dots x_{q-1}^{j_{q-1}} x_{q+1}^{j_{q+1}}\dots x_{n-1}^{j_{n-1}}=f_{n-1}({\widehat{j}},{\widehat  x})$, where the integers $\widehat 
 j$ and $\widehat x$ are defined by their binary expansions $(j_0,\dots,j_{q-1},j_{q+1},\dots,j_{n-1})$ and $(x_0,\dots,x_{q-1},x_{q+1},\dots,x_{n-1})$, respectively. The evaluations of polynomials $f_{n-1}({\widehat j},{\widehat x})$ over elements $\widehat x\in [2^{n-1}]$ for $\widehat j\in [2^{n-1}]$ give the matrix $G^{\otimes (n-1)}$. Therefore, the condition $(G^{\otimes n})_{I,J}=G^{\otimes (n-1)}$ is satisfied. 
\end{proof}

\begin{lemma}
\label{lem:recSetsIJ}
    Given any set $Q\subset [n]$ and the corresponding 
    \begin{equation}
    S(Q)\triangleq \Big\{ 
    j\in [2^n] \;|\; j_Q=\mathbf{1}\Big\},
    \label{eq:SQ}
    \end{equation}
    where $\mathbf{1}\triangleq(1,\dots,1)$, and $j_Q=\mathbf{1}$ means that $j_q=1$ for all $q\in Q$. 
    Then sets $I=J=S(Q)$ satisfy the conditions $(G^{\otimes n})_{I,J}=G^{\otimes (n-|Q|)}$ and $(G^{\otimes n})_{\overline I,J}=\mathbf{0}$.   
\end{lemma}

\begin{proof}
When $|Q|=0$, we have $I=J=S(Q)=[2^n]$ and therefore $(G^{\otimes n})_{I,J}=G^{\otimes (n-|Q|)}=G^{\otimes n}$ and $\overline I=\emptyset$. So, the statement holds for $|Q|=0$. When $|Q|=1$, Lemma \ref{lem:recSetsIJ} reduces to Lemma \ref{lem:setsIJ}. We further proceed by induction. {Induction hypothesis: assume that the statement holds for set $Q$, 
i.e., $(G^{\otimes n})_{S(Q),S(Q)}=G^{\otimes (n-|Q|)}$ and $(G^{\otimes n})_{[2^n]\setminus S(Q),S(Q)}=\mathbf 0$. 
Let us show that the statement also holds for set $\widehat Q \triangleq Q\cup \{q\}$ with any $q\in [n]\setminus Q$. 
That is, let us show that set $\widehat Q$ 
satisfies the conditions $(G^{\otimes n})_{S(\widehat Q),S(\widehat Q)}=G^{\otimes (n-|\widehat Q|)}$ and $(G^{\otimes n})_{[2^n]\setminus S(\widehat Q),S(\widehat Q)}=\mathbf{0}$.} 

{Observe that $S(\widehat Q)=S(Q\cup\{q\})=\{ 
    j\in [2^n] \;|\; j_Q=\mathbf{1}, j_q=1 \}=S(Q)\cap S(q)$. Then we represent $(G^{\otimes n})_{S(\widehat Q),S(\widehat Q)}=(G^{\otimes n})_{S(Q\cup\{q\}),S(Q\cup\{q\})}=\left((G^{\otimes n})_{S(Q),S(Q)}\right)_{\widehat S(\widehat q),\widehat S(\widehat q)}=\left(G^{\otimes n-|Q|}\right)_{\widehat S(\widehat q),\widehat S(\widehat q)}=\left(G^{\otimes \widehat n}\right)_{\widehat S(\widehat q),\widehat S(\widehat q)}$, where $\widehat n\triangleq n-|Q|$, $\widehat q\triangleq q-|\{t\in Q \,|\, t<q\}|$ and $\widehat S(\widehat q)\triangleq \{ j\in [2^{\widehat n}]\,|\, j_{\widehat q}=1]\}$. By applying Lemma \ref{lem:setsIJ} to 
$G^{\otimes \widehat n}$ and $\widehat q$,
we obtain that $\left(G^{\otimes \widehat n}\right)_{\widehat S(\widehat q),\widehat S(\widehat q)}=G^{\otimes (\widehat n-1)}=G^{\otimes (n-|\widehat Q|)}$ and $\left(G^{\otimes \widehat n}\right)_{[2^{\widehat n}]\setminus \widehat S(\widehat q),\widehat S(\widehat q)}=\mathbf{0}$. Thus, we have proved that $(G^{\otimes n})_{S(\widehat Q),S(\widehat Q)}=G^{\otimes (n-|\widehat Q|)}$. 
} 

{Now it remains to prove that $(G^{\otimes n})_{[2^n]\setminus S(\widehat Q),S(\widehat Q)}=\mathbf{0}$. For this, we use the recently shown property $\left(G^{\otimes \widehat n}\right)_{[2^{\widehat n}]\setminus \widehat S(\widehat q),\widehat S(\widehat q)}=\mathbf{0}$ and the following set properties $S(\widehat Q)=S(Q)\cap S(q)$ and $[2^n]\setminus S(\widehat Q)=[2^n]\setminus (S(Q)\cap S(q))=([2^n]\setminus S(Q))\cup (S(Q)\cap([2^n]\setminus S(q)))$. By applying the set properties, we split matrix $(G^{\otimes n})_{[2^n]\setminus S(\widehat Q),S(\widehat Q)}$ into two submatrices $(G^{\otimes n})_{[2^n]\setminus S(Q),S(\widehat Q)}$ and $(G^{\otimes n})_{S(Q)\cap([2^n]\setminus S(q)),S(\widehat Q)}$. By applying the recently shown property, we have that the second submatrix $(G^{\otimes n})_{S(Q)\cap([2^n]\setminus S(q)),S(\widehat Q)}=(G^{\otimes \widehat n})_{[2^{\widehat n}]\setminus \widehat S(\widehat q),\widehat S(\widehat q)}=\mathbf{0}$. Due to the induction hypothesis, i.e., $(G^{\otimes n})_{[2^n]\setminus S(Q),S(Q)}=\mathbf{0}$, 
we obtain that the first submatrix $(G^{\otimes n})_{[2^n]\setminus S(Q),S(\widehat Q)}=(G^{\otimes n})_{[2^n]\setminus S(Q),S(Q)\cap S(q)}=\mathbf{0}$. Thus, we have shown that $(G^{\otimes n})_{[2^n]\setminus S(\widehat Q),S(\widehat Q)}=\mathbf{0}$. This concludes the proof.}  
\end{proof}

 The following theorem summarizes Lemmas \ref{lem:H}--\ref{lem:recSetsIJ}.
 
\begin{theorem}
\label{theor:H}
Let $m\in [2^n]\setminus \{0\} $ and set $Q\subset [n]$ satisfy 
$m_Q=\mathbf{1}$ and $i_Q\neq\mathbf{1}$ for all $i\in\mA^{(m-1)}$. 
Then 
$$H(U_m|Y_{[2^n]},U_{\mF^{(m-1)}}) \leq H_{n-|Q|,|S(Q)\cap [m]|}.$$
\end{theorem}

\begin{proof}
Let us show that such $Q$ defines sets $I=J=S(Q)$ meeting all conditions of Lemma \ref{lem:H}. 
By Lemma \ref{lem:recSetsIJ}, the sets $I=J=S(Q)$ with $\widetilde n=n-|Q|$ satisfy conditions 1--3 of Lemma \ref{lem:H}. 
It follows from the restriction $m_Q=\mathbf{1}$ and Eq. \eqref{eq:SQ} that 
$m\in S(Q)$, and therefore condition 4 of Lemma \ref{lem:H} is satisfied. 
Due to the restriction $i_Q\neq\mathbf{1}$ for all $i\in\mA^{(m-1)}=[m]\setminus\mF^{(m-1)}$, we have $\{ i\in [m] \, | \, i_Q=\mathbf{1} \}\subseteq \mF^{(m-1)}$. It follows from Eq. \eqref{eq:SQ} that $\{ i\in [m] \, | \, i_Q=\mathbf{1} \}=S(Q)\cap [m]$. Thus, we obtain 
$S(Q)\cap [m]\subseteq \mF^{(m-1)}$, which means that condition 5 of Lemma \ref{lem:H} is satisfied.  Therefore, by substituting $\widetilde n=n-|Q|$ and $\widetilde m=|S(Q)\cap [m]|$ in Eq. \eqref{eq:lemH}, we obtain $H(U_m|Y_{[2^n]},U_{\mF^{(m-1)}}) \leq H_{n-|Q|,|S(Q)\cap [m]| }$.     
\end{proof}

According to Theorem \ref{theor:H}, there always exists at least one set $Q$ if $m>0$. Specifically, {for an $m\in [2^n]\setminus\{0\}$,} it is easy to see that $Q=\{ t\in [n] \, | \, m_t=1 \}$ satisfies the condition $m_Q=\mathbf{1}$, as well as $i_Q\neq\mathbf{1}$ for all $i\in \mA^{(m-1)}$ since $i_Q=\mathbf{1}$ may be true only for $i\geq m$. In this case, $|Q|=\wt(m)$ and $|S(Q)\cap [m]|=0$, leading to a simple upper bound $H(U_m|Y_{[2^n]},U_{\mF^{(m-1)}})\leq H_{n-\wt(m),0}$.

Note that there could exist several sets $Q$ satisfying conditions of Theorem \ref{theor:H}. It is desirable to find set $Q$ that provides the tightest upper bound $H(U_m|Y_{[2^n]},U_{\mF^{(m-1)}}) \leq H_{n-|Q|,|S(Q)\cap [m]| }$. This requires to solve the following optimization problem:  
\begin{equation}
\label{eq:Q*}
Q^*=
\min_{Q\in \mathbb{Q}}
H_{n-|Q|,|S(Q)\cap [m]|},
\end{equation}
\begin{equation}
\label{eq:boldQ}
\mathbb{Q}\triangleq \{ Q\subset [n] \, | \, m_Q=\mathbf{1}, \forall i\in\mA^{(m-1)}\; i_Q\neq\mathbf{1}\}.
\end{equation}
The number of sets $Q$ to consider is upper bounded by $2^n$, i.e., by the code length $N=2^n$. Note that the condition $m_Q=\mathbf{1}$ reduces this number to $2^{\wt(m)}$, where $\wt(m)$ is the Hamming weight of the binary expansion of $m$. Since the cardinality of the set $\mA^{(m-1)}$ is upper bounded by $m$, we conclude that the time complexity of finding $Q^*$ scales as $O(2^{\wt(m)}\cdot m)$, assuming that the bit-channel entropies are pre-computed. Note that $\wt(m)\leq n$ and $m< 2^n$.     

The following lemma simplifies the search for $Q^*$ by showing that $H_{n-|Q|,|S(Q)\cap [m]|}$ cannot be decreased 
by including additional elements into $Q$.

\begin{lemma}
\label{lem:quality}
{For any sets $Q'\subset [n]$ and $Q\subset Q'$}, 
$$H_{n-|Q|,|S(Q)\cap [m]|}\leq H_{n-|Q'|,|S(Q')\cap [m]|}.$$
\end{lemma}

\begin{proof}
Using the notation of Lemma \ref{lem:H} and Theorem \ref{theor:H}, $\widetilde n=n-|Q|$ and $\widetilde m = |S(Q)\cap [m]|$. Let us denote $\delta\triangleq |Q'\setminus Q|$ and $\widetilde m'\triangleq |S(Q')\cap [m]|$. 
The entropy $H_{\widetilde n,\widetilde m}$ corresponds to the $\widetilde m$-th bit-channel of the {polar} transformation $G^{\otimes \widetilde n}$, denoted by $W_{\widetilde n, \widetilde m}$, while the entropy $H_{\widetilde n',\widetilde m'}$ characterizes the $\widetilde m'$-th bit-channel of the {polar} transformation $G^{\otimes (\widetilde n-\delta)}$, denoted by $W_{\widetilde n-\delta, \widetilde m'}$. 
It follows from Eq. \eqref{eq:SQ} that the binary expansion of $\widetilde m'$ can be obtained from the binary expansion of $\widetilde m$ by deleting $\delta$ bits equal to $1$, whose indices are defined by $Q'\setminus Q$. Therefore, by deleting the polarization layers with these indices from the {polar} transformation $G^{\otimes \widetilde n}$, the bit-channel $W_{\widetilde n, \widetilde m}$ can be transformed into 
$W_{\widetilde n-\delta, \widetilde m'}$.  
Since the deleted bits are all equal to $1$, $W_{\widetilde n-\delta, \widetilde m'}$ has lower symmetric capacity and higher entropy than 
$W_{\widetilde n, \widetilde m}$ as follows from \cite[Section III]{arikan2009channel}. 
\end{proof}

By Lemma \ref{lem:quality}, it suffices to explore only a subset of $\mathbb{Q}$ to find $Q^*$. Specifically, 
\begin{equation}
\label{eq:Q*new}
Q^*=\min_{Q\in \widehat{\mathbb{Q}}}
H_{n-|Q|,|S(Q)\cap [m]|},
\end{equation}
$$\widehat{\mathbb{Q}}\triangleq \{ Q\in\mathbb{Q} \, | \, \forall q\in Q \quad Q\setminus\{q\}\notin\mathbb{Q} 
\}.$$
{Thus, $\widehat{\mathbb{Q}}$ is the subset of $\mathbb{Q}$ such that there does not exist any distinct pair $Q, \tilde{Q}\in\widehat{\mathbb{Q}}$ with $Q\subset\tilde{Q}$.} {Set} $Q^*$ should have a low cardinality compared to the other $Q\in \mathbb{Q}$. We further propose a low-complexity greedy 
approach aiming to find $Q\in \mathbb{Q}$ with the lowest cardinality. 
Algorithm \ref{alg:Q} specifies the proposed approach in which we initialize set $Q$ {as} the empty set and then add an element $q^*$ to $Q$ at each iteration of the \textbf{while} loop until the condition $Q\in\mathbb{Q}$ is satisfied.  {At line 3, we calculate set $M\triangleq \{ t\in [n] \,|\,m_t=1\}$ containing all possible $q^*$. Note that all $q^*$ 
must belong to
the set $M$ to ensure that $Q$ satisfies to the condition $m_Q=\mathbf{1}$ 
in Eq. \eqref{eq:boldQ}. At line 4, we initialize set $\Lambda$, which is further iteratively updated at line 9. At each iteration of the \textbf{while} loop, the updated set $\Lambda$ consists of all $i\in\mA^{(m-1)}$ violating the condition $i_Q\neq\mathbf{1}$ for the current set $Q$ in Eq. \eqref{eq:boldQ}, i.e., $\Lambda=\{ i\in\mA^{(m-1)} \,|\, i_Q=\mathbf{1}\}$. Each $q^*$ is calculated at line 7 as the element $q\in M\setminus Q$ minimizing the cardinality of the updated set $\Lambda$, assigned at line 9.} 
The number of iterations in the \textbf{while} loop is upper bounded by $|M|$ since $|M|$ iterations result in $Q=M$ and such $Q$ satisfies $\forall i\in\mA^{(m-1)}$ $i_Q\neq\mathbf{1}$ as explained right below the proof of Theorem \ref{theor:H}. 

\begin{algorithm2e}
\small
\setstretch{0.8}
\caption{
\small Greedy approach to optimize set $Q$ 
\label{alg:Q}}
\DontPrintSemicolon
\SetKwFor{loop}{loop}{}{end}
\textbf{ConstructSetQ}$(n, m, \mA^{(m-1)})$ 

\Begin{
$M\leftarrow \{ t\in [n] \, | \, m_t=1\}$\;
$Q\leftarrow\emptyset$\;
$\Lambda\leftarrow \mA^{(m-1)}$\;
\While{$|\Lambda|>0$}{
    $\displaystyle q^*\leftarrow\argmin_{q\in M\setminus Q} |\{ {i}\in\Lambda          \; |\; {i}_q=1 \}|$\; 
    $Q\leftarrow Q\cup\{ q^*\}$\;
    $\Lambda\leftarrow \{ i\in\Lambda          \; |\; i_{q^*}=1 \}$\;
    }
     \Return $Q$
}
\end{algorithm2e}

Although Algorithm \ref{alg:Q} does not guarantee optimality, the resulting set $Q$ is typically equal to $Q^*$. 
The worst-case time complexity of Algorithm \ref{alg:Q} {scales as} $O(\wt(m)^2\cdot m)$ since the maximum number of the \textbf{while} loop iterations is $|M|=\wt(m)$ and the complexity of each iteration is dominated by line 7, whose complexity is upper bounded by $|M|\cdot m$ {bit comparisons and bit counting}. {From the perspective of practical implementation, the complexity of operations over bits is low, leading to a low practical complexity of Algorithm \ref{alg:Q}. For example, there are $\wt(m)^2\cdot m\leq 9^2\cdot 512$ bit operations for the code length $N=512$. A single floating-point number typically has $32$ or $64$ bits, which is comparable with $\wt(m)^2\leq 9^2$ for $N=512$.}   

The proposed \textit{\textbf{tightened lower bound}} on $\bar D_m$ can be computed as follows   
\begin{equation}
    \bar D^{\mathrm{tight}}_{m}=\begin{cases} \bar D^{\mathrm{tight}}_{m-1}+H_{n,m}, & m\in\mA, \\ \max(0, \bar D^{\mathrm{tight}}_{m-1}-\\ \quad (H_{n-|Q(m)|,|S(Q(m))\cap [m]|} - H_{n,m})), & m\in\mF,\end{cases}
    \label{eq:tightLower}
\end{equation}
where $m\in [2^n]$, $\bar D^{\mathrm{tight}}_{-1}=0$, and $Q(m)$ means the set $Q$ calculated for a particular $m$ by Algorithm \ref{alg:Q} or Eq. \eqref{eq:Q*new}. We used Algorithm \ref{alg:Q} to produce numerical results for Section \ref{sNumerical}. 

{The computation of $\bar D^{\mathrm{tight}}_{m}$ in Eq. \eqref{eq:tightLower} involves the calculation of $Q(m)$.}
Therefore, the overall time complexity {in terms of the number of bit operations for} computing $\bar D^{\mathrm{tight}}_{0},\dots,\bar D^{\mathrm{tight}}_{2^n-1}$ using Algorithm \ref{alg:Q} scales as $O(\sum_{m=0}^{2^n-1} \wt(m)^2\cdot m)$, which is upper bounded by $O(n^2\cdot 2^{2n})=O(\log(N)^2\cdot N^2)$. {The complexity in terms of the number of floating-point operations scales as $O(N)$, as follows from Eq. \eqref{eq:tightLower}.} 
{Note that the number of floating-point operations for computing $\bar D^{\mathrm{low}}_{0},\dots,\bar D^{\mathrm{low}}_{2^n-1}$ in Eq. \eqref{eq:lower} also scales as $O(N)$, but Eq. \eqref{eq:lower} does not require bit operations. 
Thus, the complexity of the proposed $\bar D^{\mathrm{tight}}_{0},\dots,\bar D^{\mathrm{tight}}_{2^n-1}$ is higher due to the necessity to perform additional bit operations. Nevertheless, we consider the complexity of the proposed bound to be acceptable for the frozen set design, since it does not exceed the complexity of the ML performance prediction described in Section \ref{sPeML}, and it is much lower than the complexity of decoding simulations. 
} 

\begin{figure}
    \centering
    \includegraphics[width=0.48\textwidth]{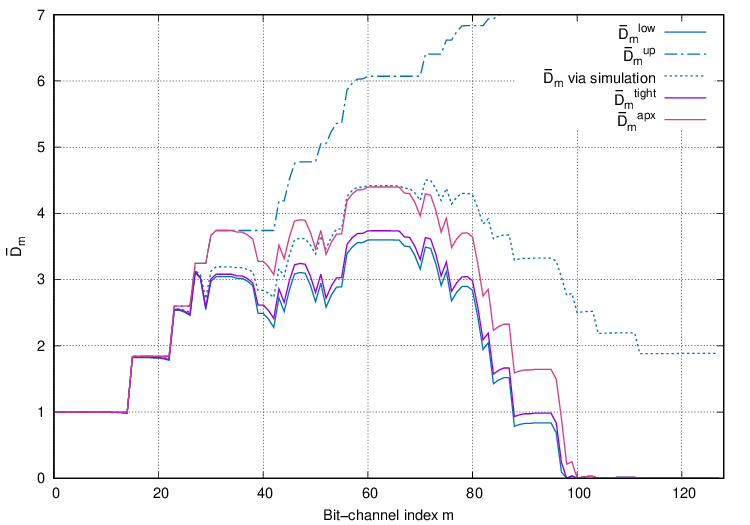} 
    \caption{{$\bar D_m$ for the $(128,64)$ code proposed in \cite{Coskun2022InfTheor}} 
    }
    \label{fig:accuracy}
\end{figure}

{Fig. \ref{fig:accuracy} illustrates the accuracy of the bounds on $\bar D_m$ in the same manner as \cite[Fig. 1]{Coskun2022InfTheor}. 
The $(128,64)$ code proposed in \cite[Section V-A]{Coskun2022InfTheor} is considered. It can be seen that the tightened lower bound $\bar D^{\mathrm{tight}}_m$ provides a noticeable improvement compared to the lower bound $\bar D^{\mathrm{low}}_m$ \cite{Coskun2022InfTheor}. In particular, $\max_{m\in [N]}\bar D^{\mathrm{tight}}_m$ is closer to the simulated $\max_{m\in [N]}\bar D_m$ than $\max_{m\in [N]}\bar D^{\mathrm{low}}_m$, where the simulated $\bar D_m$ is from \cite[Fig. 1]{Coskun2022InfTheor}. 
Note that $\bar D^{\mathrm{up}}_m$ in Fig. \ref{fig:accuracy} is the upper bound \cite[Eq. (6b)]{Coskun2022InfTheor}. 
We have also included the approximate bound $\bar D^{\mathrm{apx}}_m$, proposed below in Section \ref{sApproximate}.} 

\subsection{Proposed Approximate Bound $\bar D^{\mathrm{apx}}_{m}$ that Combines the Tightened Lower Bound and Upper Bound}
\label{sApproximate}


Reference \cite{Coskun2022InfTheor} proposed several frozen sets without 
the first bit-channel, 
i.e., the less reliable bit-channel is used to transfer information bits in \cite{Coskun2022InfTheor}. 
To the best of the authors' knowledge, such frozen sets have not been used before. 
This motivated us to investigate why the first bit-channel is not frozen in \cite{Coskun2022InfTheor}.  

Observe that \cite[Remark 2]{Coskun2022InfTheor} suggests the frozen set design criterion $\log_2(L)\geq \bar D_m$ and the usage of $\bar D^{\mathrm{low}}_{m}$ as a proxy for $\bar D_m$, where $L$ is the target SCL list size. 
Therefore, 
we explore the gap between $\bar D_{m^*}$ and $\bar D^{\mathrm{low}}_{m^*}$, where $m^*\triangleq \argmax_{m\in [2^n]} \bar D^{\mathrm{low}}_{m}$. 
We further provide expressions in terms of the information set $\mA=[2^n]\setminus\mF$. Note that $\bar D_{m}=\bar D^{\mathrm{low}}_{m}=0$ holds for $0\leq m<\min(\mA)$ and all information sets $\mA$.
{Let $\mathbb{M}\triangleq \{\min(\mA),\dots,m^*\}$. It follows from Eq. \eqref{eq:lower} that $\bar D^{\mathrm{low}}_{m^*}=\sum_{m\in\mathbb{M}} (\bar D^{\mathrm{low}}_{m}-\bar D^{\mathrm{low}}_{m-1})\geq\sum_{i\in\mA^{(m^*)}}H_{n,i}-\sum_{i\in\mathbb{M}\cap\mF^{(m^*)}}(1-H_{n,i})$. The inequality becomes an equality for the information sets with $\bar D^{\mathrm{low}}_{m}>0$ for $m\in\mathbb{M}$.  
The property $\bar D^{\mathrm{low}}_{m}>0$, $m\in\mathbb{M}$, generally holds for information sets designed for moderate-to-large $L$. The corresponding examples can be found in \cite[Figs. 1 and 3]{Coskun2022InfTheor}.} 
{
It follows from Section \ref{sCoskunDeriv} that 
$\bar D_{m^*}=\sum_{m\in\mathbb{M}} (\bar D_{m}-\bar D_{m-1})
=\sum_{m\in\mA^{(m^*)}} H_{n,m}+\sum_{m\in\mathbb{M}\cap\mF^{(m^*)}} (H_{n,m}-H(U_m|Y_{[N]}, U_{\mF^{(m-1)}}))$. 
Besides, observe that $\mathbb{M}\cap\mF^{(m^*)}=\mathbb{M}\cap\mF=\mathbb{M}\setminus\mA$.
Thus, we obtain that the gap $g(\mA)\triangleq\bar D_{m^*} - \bar D^{\mathrm{low}}_{m^*} \leq 
\sum_{m\in \mathbb{M}\setminus\mA 
} (1- H(U_m|Y_{[N]}, U_{\mF^{(m-1)}}))$. The inequality becomes an equality for the information sets with $\bar D^{\mathrm{low}}_{m}>0$, $m\in\mathbb{M}$, i.e., the gap $g(\mA)=\sum_{m\in \mathbb{M}\setminus\mA 
} (1- H(U_m|Y_{[N]}, U_{\mF^{(m-1)}}))$. 
}
Thus, the gap $g(\mA)$ substantially depends on {the cardinality of set $\mathbb{M}=\{\min(\mA),\dots,m^*\}$, which depends on} the index of the lowest information bit $\min(\mA)$. The lower $\min(\mA)$, the higher $g(\mA)$. In particular, when $\min(\mA)=0$ as for the information sets proposed in \cite{Coskun2022InfTheor}\footnote{In \cite{Coskun2022InfTheor}, the enumeration starts from 1, and therefore the lowest information bit index is equal to 1.}, the gap $g(\mA)$ is especially {high} due to {a large} number of terms  in $\sum_{m\in \{\min(\mA)=0,\dots,m^*\}\setminus\mA 
} (1- H(U_m|Y_{[N]}, U_{\mF^{(m-1)}}))$. 
That is, the lower bound $\bar D^{\mathrm{low}}_{m^*}$ especially underestimates $\bar D_{m^*}$ when $\min(\mA)$ is close to zero.
This means that the code design criterion  $\max_{m\in [2^n]} \bar D^{\mathrm{low}}_{m}$ 
gives preference to information sets with 
a very low $\min(\mA)$. That is why the information sets constructed in \cite{Coskun2022InfTheor} have $\min(\mA)=0$. The replacement of $\bar D^{\mathrm{low}}_{m}$ by our tightened lower bound $\bar D^{\mathrm{tight}}_{m}$ partially solves the problem by reducing the gap $\bar D_{m^*} - \bar D^{\mathrm{tight}}_{m^*} 
\leq \bar D_{m^*} - \bar D^{\mathrm{low}}_{m^*}$. 

To eliminate the bias towards the information sets $\mA$ having low $\min(\mA)$,  
we propose to combine our tight lower bound $\bar D^{\mathrm{tight}}_{m}$ with the upper bound $\bar D^{\mathrm{up}}_{m}\triangleq \sum_{m\in\mA^{(m)}} H_{n,m}$ from \cite[Eq. (6b)]{Coskun2022InfTheor} as follows
\begin{equation}
    \bar D^{\mathrm{apx}}_{m}=\begin{cases} \bar D^{\mathrm{apx}}_{m-1}+H_{n,m}, & m\in\mA, \\
    \bar D^{\mathrm{apx}}_{m-1}, & m\in\mF\cap [\lambda], \\
    \max(0, \bar D^{\mathrm{apx}}_{m-1} -\\ \quad (H_{n-|Q(m)|,|S(Q(m))\cap [m]|} -\\ \quad H_{n,m})), & m\in\mF\setminus [\lambda],\end{cases}
    \label{eq:tightApprox}
\end{equation}
where $m\in [2^n]$, $\bar D^{\mathrm{apx}}_{-1}=0$, and 
$\lambda$ is an integer threshold.  According to \eqref{eq:tightApprox}, $\bar D^{\mathrm{apx}}_{m}=\bar D^{\mathrm{up}}_{m}$ for $m\in [\lambda]$. Thus, for the frozen bits with low indices $m\in \mF\cap [\lambda]$, we use $\bar D^{\mathrm{apx}}_{m}-\bar D^{\mathrm{apx}}_{m-1}=\bar D^{\mathrm{up}}_{m}-\bar D^{\mathrm{up}}_{m-1}=0$. For the remaining frozen bits $m\in \mF\setminus [\lambda]$, we employ $\bar D^{\mathrm{apx}}_{m}-\bar D^{\mathrm{apx}}_{m-1}=\bar D^{\mathrm{tight}}_{m}-\bar D^{\mathrm{tight}}_{m-1}$. Note that all considered bounds process the information bits $m\in\mA$ in the same way: $\bar D^{\mathrm{apx}}_{m}-\bar D^{\mathrm{apx}}_{m-1}=\bar D^{\mathrm{low}}_{m}-\bar D^{\mathrm{low}}_{m-1}=\bar D^{\mathrm{tight}}_{m}-\bar D^{\mathrm{tight}}_{m-1}=\bar D^{\mathrm{up}}_{m}-\bar D^{\mathrm{up}}_{m-1}=H_{n,m}$. Obviously, $\bar D^{\mathrm{low}}_{m}\leq\bar D^{\mathrm{tight}}_{m}\leq \bar D^{\mathrm{apx}}_{m}\leq \bar D^{\mathrm{up}}_{m}$. If $\min(\mA)\geq \lambda$, then $\bar D^{\mathrm{apx}}_{m}=\bar D^{\mathrm{tight}}_{m}$ for all $m$. 
It can be seen that the gap between $\bar D_{m^*}$ and $\bar D^{\mathrm{apx}}_{m^*}$ depends on $\lambda$ instead of $\min(\mA)$, since $\bar D_{m^*}-\bar D^{\mathrm{apx}}_{m^*}\leq\sum_{m\in [\lambda]\setminus\mA} (H(U_m|Y_{[N]}, U_{[m]})-H(U_m|Y_{[N]}, U_{\mF^{(m-1)}})) + \sum_{m\in \{ \lambda,\dots, m^*\}\setminus\mA} (H_{n-|Q(m)|,|S(Q(m))\cap [m]|}- H(U_m|Y_{[N]}, U_{\mF^{(m-1)}}))$. 

It is very important that the value of $\lambda$ is the same for all information sets being compared during the code design process to enable a fair comparison. 
{It is desirable to set $\lambda$ close to 
$\max_{\mA}(\min(\mA))$,  where the maximization is performed over the
information sets $\mA$ being compared, assuming that the obviously bad information sets, e.g., $\{N-K,\dots,N-1\}$, are eliminated from the consideration. Note that 
$\lambda$ substantially different from $\max_{\mA}(\min(\mA))$ is unsuitable
for two reasons: (i) $\bar D^{\mathrm{apx}}_{m}$ reduces to $\bar D^{\mathrm{tight}}_{m}$ for $\mA$ with $\min(\mA)\geq \lambda$, i.e., no effect from such $\lambda$, and (ii) if $\lambda$ is substantially higher than $\max_{\mA}(\min(\mA))$, then the upper bound $\bar D^{\mathrm{up}}_{m}$ is applied to unnecessary many bits, i.e., such $\lambda$ makes $\bar D^{\mathrm{apx}}_{m}$ close to $\bar D^{\mathrm{up}}_{m}$.
We further focus on the problem of prediction $\max_{\mA}(\min(\mA))$ for fixed code parameters $(N,K)$ and varying decoding list size $L$. 
Observe that the SCL decoder with $L>1$ delays making decisions on the values of information bits in contrast with the SCL decoder with $L=1$, i.e., the SC decoder. Therefore, codes designed for SCL decoding with $L>1$ are likely to have lower indices of information bits than codes designed for SC decoding. 
For example, a typical CRC-aided polar code has an information set constructed by selecting the $K+c$ most reliable bit-channels and then eliminating $c$ bit-channels with the highest indices, where $c$ is the CRC length.  
Note that 
 the CRC-aided polar code with $c=0$ is the conventional polar code \cite{arikan2009channel}, because its information set consists of the most reliable bit-channels. This code with $c=0$ suits well for the SCL decoder with $L=1$. The higher $c$, the larger $L$ is typically used in the SCL decoder. It follows from the definition of the CRC-aided polar code information set that its $\min(\mA)$ weakly decreases with increasing $c$. Similarly, the average value of the information bit indices, i.e., $\frac{1}{K}\sum_{i\in\mA}i$, also decreases with increasing $c$. 
For $(128,64)$ CRC-aided polar codes\footnote{{For $(N,K)=(128,64)$, we calculated the bit-channel reliabilities using the Gaussian approximation
\cite{trifonov2012efficient} for AWGN, BPSK and $E_b/N_0=3.5$ dB.}}, if the CRC length $c$ increases as (0,1,...,14), then $\frac{1}{K}\sum_{i\in\mA}i$ decreases as (89.1, 88.4, 88, 86.5, 85.2, 84.5, 83.5, 82.8, 81.3, 80.7, 79.7, 79, 78, 77.3, 75.9)
 and $\min(\mA)$ weakly decreases as (30, 30, 30, 29, 29, 29, 29, 29, 27, 27, 27, 27, 27, 27, 23). The $(128,64)$ Reed-Muller code has $\frac{1}{K}\sum_{i\in\mA}i=83.3$ and $\min(\mA)=15$. 
 The $(128,64)$ code \cite{Coskun2022InfTheor} has $\frac{1}{K}\sum_{i\in\mA}i=83.2$ and $\min(\mA)=0$.
 The $(128,64)$ systematic PAC code \cite{Tonnellier2021SystPAC} has $\frac{1}{K}\sum_{i\in\mA}i=80.4$ and $\min(\mA)=15$.
 The $(128,64)$ extended BCH code has $\frac{1}{K}\sum_{i\in\mA}i=72.8$ and $\min(\mA)=7$.
 Thus, $\min(\mA')$ can be used as a proxy for $\max_{\mA}(\min(\mA))$, where $\mA'$ is the information set consisting of the most reliable bit-channels. Therefore, we recommend setting $\lambda$ close to $\min(\mA')$.}
{Since} $\min(\mA')$ is equal to $30$ and $95$ for the code parameters $(128,64)$ and $(512,256)$, respectively,  
{we} use $\lambda=2^5=32$ for $(128,64)$ and $\lambda=3\cdot 2^5=96$ for $(512,256)$. 

\subsection{Frozen Set Optimization} 
\label{sOptimization}


In this section, we consider the frozen set optimization problem 
with two objectives: minimize the decoding complexity characterized by $\bar D_{\mathrm{apx}}\triangleq \max_{m\in [2^n]} \bar D^{\mathrm{apx}}_{m}$, defined by Eq. \eqref{eq:tightApprox}, and minimize the ML decoding error probability estimate $\widetilde{P}_{\mathrm{ML}}$, computed as in Section \ref{sPeML}. 
Since these two objectives are conflicting, we are interested in constructing frozen sets leading to codes with various complexity-performance tradeoffs. 
The best tradeoffs are provided by the \textit{\textbf{Pareto front}}, which
is the set of all non-dominated $(\bar D_{\mathrm{apx}}, \widetilde{P}_{\mathrm{ML}})$, 
i.e., the Pareto front consists of pairs $(\bar D_{\mathrm{apx}}, \widetilde{P}_{\mathrm{ML}})$ such that all other pairs $(\bar D_{\mathrm{apx}}', \widetilde{P}_{\mathrm{ML}}')$ satisfy $\bar D_{\mathrm{apx}}'>\bar D_{\mathrm{apx}}$ or $\widetilde{P}_{\mathrm{ML}}'>\widetilde{P}_{\mathrm{ML}}$, where the computations are performed for fixed code length and rate. 

\subsubsection{Optimization using the Genetic Algorithm GenAlgT}
\label{sGenAlgT}
The computational complexity of finding the exact Pareto front is huge since there are plenty of frozen sets to consider. Therefore, we find an approximate Pareto front using 
a variation of the genetic algorithm \cite{Elkelesh2019} with the hash table \cite{Zhou2021GenAlgHash} to reduce time complexity and with the elimination of identical candidates from the population to preserve diversity. 
Since the genetic algorithm \cite{Elkelesh2019} has only one objective of minimizing the decoding FER/BER, 
we need to adjust it. Specifically, we modify the genetic algorithm so that it solves the constrained minimization problem: minimize $\widetilde{P}_{\mathrm{ML}}$ subject to the constraint $\bar D_{\mathrm{apx}}\leq T_D$, where the threshold $T_D$ is an input parameter of the genetic algorithm. 
To ensure that the frozen set population satisfies this constraint, 
we discard frozen sets violating this constraint from the initial population and from the crossover output. Besides, we allow the mutation operation to swap a frozen bit and a non-frozen bit only when this does not lead to the constraint violation. 

\begin{algorithm2e}
\small
\setstretch{0.8}
\caption{
\small Genetic algorithm
\label{alg:GenAlgT}}
\DontPrintSemicolon
\SetKwRepeat{Do}{do}{while}
\SetKwFunction{FInitializePopulation}{InitializePopulation}
\SetKwFunction{FPrunePopulation}{PrunePopulation}
\SetKwFunction{FExtendPopulation}{ExtendPopulation}
\SetKwFunction{FRandom}{Random}
\SetKwFunction{FMutation}{Mutation}
\SetKwFunction{FCrossover}{Crossover}
\SetKwFunction{FBest}{Best}
\SetKwFunction{FEvolve}{Evolve}
{\textbf{GenAlgT}$(N,K,T_{\mathrm{POP}},T_D,\theta)$} 

{\Begin{
    $\mathbf{P}\leftarrow \FInitializePopulation(N,K)$\;
    $\mathbf{P}\leftarrow \{ \mF\in\mathbf{P}\;|\; \bar D_{\mathrm{apx}}(\mF) \leq T_D\}$\; $\mF^*\leftarrow\FEvolve(N,K,T_{\mathrm{POP}},T_D,\theta,\mathbf{P}, [N])$\;
    \Return $\mF^*$\;
}}

{\textbf{Evolve}$(N,K,T_{\mathrm{POP}},T_D,\theta, \mathbf{P}_{\mathrm{INIT}}, \mathcal{Z})$} 

{\Begin{
$\mathbf{P}\leftarrow \FPrunePopulation(\mathbf{P}_{\mathrm{INIT}},T_{\mathrm{POP}})$\;
\For{$t=0,\dots,\theta-1$}{
    $\mathbf{P}\leftarrow \FExtendPopulation(\mathbf{P},T_D,\mathcal{Z})$\; 
    $\mathbf{P}\leftarrow \FPrunePopulation(\mathbf{P},T_{\mathrm{POP}})$\;
    \textbf{if} $\widetilde{P}_{\mathrm{ML}}(\FBest(\mathbf{P}))$ has reduced \textbf{then} $t\leftarrow 0$\; 
    }
     \Return $\FBest(\mathbf{P})$ 
}}
{\textbf{Best}$(\mathbf{P})$}

{\Begin{
    \Return $\argmin_{\mF\in\mathbf{P}} \widetilde{P}_{\mathrm{ML}}(\mF)$\; 
}}
{\textbf{PrunePopulation}$(\mathbf{P},T_{\mathrm{POP}})$}

{\Begin{
    \While{$|\mathbf{P}|>T_{\mathrm{POP}}$}{
        $\mathbf{P}\leftarrow \mathbf{P}\setminus \{ \argmax_{\mF\in\mathbf{P}}\widetilde{P}_{\mathrm{ML}}(\mF) \}$\;
    }
    \Return $\mathbf{P}$\; 
}}
{\textbf{ExtendPopulation}$(\mathbf{P},T_D,\mathcal{Z})$}

{\Begin{
    \Return $\mathbf{P} \cup \{ \FMutation(\mF,T_D,\mathcal{Z})\; |\; \mF\in\mathbf{P}\} \cup \{ \FCrossover(\mF,\mF',T_D,\mathcal{Z})\; |\; \mF,\mF'\in\mathbf{P},\; \mF\neq \mF'\}$\;
}}
{\textbf{Mutation}$(\mF,T_D,\mathcal{Z})$}
    
{\Begin{
\Do{ $\bar D_{\mathrm{apx}}(\mF') > T_D$
}{$\mF'\leftarrow \mF\cup\{\FRandom(([N]\setminus\mF)\cap \mathcal{Z})\}\setminus\{\FRandom(\mF\cap \mathcal{Z})\}$\;} 
    \Return $\mF'$
}}
{\textbf{Crossover}$(\mF^{(0)},\mF^{(1)},T_D,\mathcal{Z})$}
    
{\Begin{
    $b\leftarrow \FRandom(\{0,1\})$\;
    $\mF'\leftarrow (\mF^{(b)}\cap[N/2])\cup (\mF^{(1-b)}\cap([N]\setminus[N/2]))$\;
    \While{$|\mF'|>|\mF^{(b)}|$}{
        $\mF'\leftarrow \mF'\setminus\{\FRandom(\mF'\cap \mathcal{Z})\}$
    }
    \While{$|\mF'|<|\mF^{(b)}|$}{
        $\mF'\leftarrow \mF'\cup\{\FRandom(([N]\setminus\mF')\cap\mathcal{Z})\}$
    }
    \textbf{if} $\bar D_{\mathrm{apx}}(\mF') \leq T_D$ \textbf{then} 
       \Return $\mF'$    
    \textbf{else}
        \Return $\emptyset$\; 
}}
\end{algorithm2e}

{The resulting genetic algorithm, referred to as \textit{\textbf{GenAlgT}}, is summarized in Algorithm \ref{alg:GenAlgT}. The input arguments of GenAlgT are the code length $N$, dimension $K$, truncated population size $T_{\mathrm{POP}}$, threshold $T_D$ for $\bar D_{\mathrm{apx}}$ and threshold $\theta$ for the number of iterations with no improvement. We set $T_{\mathrm{POP}}=5$ as in \cite{Elkelesh2019}. The implicit input arguments are $E_b/N_0$ for $\widetilde{P}_{\mathrm{ML}}$ computation\footnote{{We use $E_b/N_0$ such that  $\widetilde{P}_{\mathrm{ML}}$ is between $10^{-6}$ and $10^{-2}$ for the corresponding precoded polar codes generated for various $T_D$. For example, we set $E_b/N_0=3.5$ dB when $(N,K)=(128,64)$, $E_b/N_0=2.0$ dB when $(N,K)=(512,256)$ and $E_b/N_0=1.5$ dB when $(N,K)=(512,128)$.}}, and the bit-channel entropies $\{H_{\widetilde n,i}\}_{\substack{i\in[2^{\widetilde n}] \\ 0\leq \widetilde n \leq n }}$ for $\bar D_{\mathrm{apx}}$ computation in Eq. \eqref{eq:tightApprox}.  At line 3 of Algorithm \ref{alg:GenAlgT}, function $InitializePopulation$ generates the initial population consisting of \begin{itemize}
\item The reliability-based frozen sets, constructed using the Gaussian approximation
\cite{trifonov2012efficient} for AWGN, BPSK and various $E_b/N_0$ with the step $0.25$ dB. 
\item The frozen sets interpolating between the Reed-Muller and reliability-based
frozen sets, where the reliability-based frozen sets were calculated using the Gaussian approximation
\cite{trifonov2012efficient} for AWGN, BPSK and $E_b/N_0$\footnote{\label{fEbN0SC}{We use $E_b/N_0$ such that the SC decoding error probability is about $0.01$. For example, we set $E_b/N_0=3.5$ dB when $(N,K)=(128,64)$, $E_b/N_0=2.75$ dB when $(N,K)=(512,256)$ and $E_b/N_0=2.25$ dB when $(N,K)=(512,128)$.}}, and then 
``interpolating'' frozen sets were generated by decreasing the number of minimum-weight information bit indices one by one. Note that the frozen set cardinality was preserved by including the next reliable information bit of higher weight whenever needed. 
\end{itemize}
At line 4, all frozen sets violating the constraint $\bar D_{\mathrm{apx}}\leq T_D$ are eliminated from the initial population. At line 5, we execute function $Evolve$. Compared to GenAlgT, $Evolve$ has two additional input arguments: the initial population $\mathbf{P}_{\mathrm{INIT}}$ and set $\mathcal{Z}$, which is the set of bit-channel indices that might be arbitrarily frozen or non-frozen. Note that $\mathcal{Z}$ is always equal to $[N]$ in Section \ref{sGenAlgT}; however, the cardinality of $\mathcal{Z}$ is substantially reduced in Sections \ref{sGenAlgTS} and \ref{sGenAlgTB}. At line 10, function $PrunePopulation$ reduces the population size to $T_{\mathrm{POP}}$ by eliminating the frozen sets with the highest $\widetilde{P}_{\mathrm{ML}}$. Then, the population iteratively evolves at lines 11--15. If there is no improvement in $\widetilde{P}_{\mathrm{ML}}(Best(\mathbf{P}))$ for $\theta$ iterations, then the process is terminated. Note that function $Best$ is specified by lines 18--21, and $\widetilde{P}_{\mathrm{ML}}(Best(\mathbf{P}))=\min_{\mF\in\mathbf{P}}\widetilde{P}_{\mathrm{ML}}(\mF)$. The number of 
iterations without improvement is counted by the variable $t$. At each iteration, function $ExtendPopulation$ adds new frozen sets to the population $\mathbf{P}$ by applying mutation to each frozen set from $\mathbf{P}$ and applying crossover to each distinct frozen set pair from $\mathbf{P}$. Function $Mutation$ randomly swaps a frozen bit and a non-frozen bit such that the resulting frozen set does not violate the constraint $\bar D_{\mathrm{apx}}\leq T_D$. Thus, $Mutation$ adds $T_{\mathrm{POP}}=5$ frozen sets to the population. Function $Random$ randomly selects an element from a given set. Function $Crossover$ merges two halves of distinct frozen sets to yield a new frozen set, whose cardinality is regulated by randomly adding or eliminating elements. Thus, $Crossover$ adds at most $T_{\mathrm{POP}}(T_{\mathrm{POP}}-1)/2=10$ frozen sets to the population.} 
{Note that the population is extended and pruned in the same way as in \cite{Elkelesh2019,Tonnellier2021SystPAC}. When $\mathcal{Z}=[N]$ and $T_D=\infty$, the functions  $Mutation$ and $Crossover$ are the same as in \cite{Elkelesh2019,Tonnellier2021SystPAC}.}



Genetic algorithms are known to be suboptimal \cite{enwiki:PremConv}, i.e, converge prematurely to local optima since genes of high-rated individuals (frozen sets) typically dominate the population. 
According to our experimental results, GenAlgT returns different outputs when run multiple times. To leverage this issue, we run GenAlgT algorithm $\rho$ times for each $T_D$. We use $\rho=5$ and consider various values of $T_D$ with the granularity $0.1$. 

\subsubsection{{Complexity} Reduction. S-Constraint and GenAlgTS} 
\label{sGenAlgTS}
Given the code length $N=2^n$ and dimension $K$, the {genetic algorithm GenAlgT performs a search over the set of} all frozen sets $\mF\subset [N]$ of cardinality $N-K$. The number of such frozen sets $\mF$
is equal to the binomial coefficient $\binom{N}{K}$, which grows rapidly with $N$ and $\min (K,N-K)$. 
We resolve this issue by introducing our constraints on the frozen set structure and incorporating them into the genetic algorithm. 

Let $r$ be the reliability sequence consisting of the bit-channel indices arranged in ascending order of their
reliabilities. We construct $r$ by using the Gaussian approximation \cite{trifonov2012efficient} for {AWGN, BPSK and $E_b/N_0$}\textsuperscript{\ref{fEbN0SC}}. Alternatively, the 5G reliability sequence \cite{polarNR2018} might be used. 
It has been shown in \cite{mondelli2014from} that both bit-channel reliabilities and index weights are of great importance when designing polar codes for SCL decoding. 
Following this direction, we characterize the closeness of a given frozen set $\mF$ to the reliability-based frozen set 
{for a given bit-channel index weight $v$ by an integer number $\alpha^\mF_v$ defined as} 
\begin{align}
\alpha^\mF_v&\triangleq \begin{cases}c_v, & 0\leq v< l^\mF,\\
\min \{ q\in [c_v] \, | \, \tau_{v,q}\notin\mF\}, & l^\mF\leq v\leq n, 
\end{cases} \nonumber \\
l^\mF&\triangleq \min_{i\in [N]\setminus\mF} \wt(i),\label{eq:l}  
\end{align}
where {$c_v\triangleq\binom{n}{v}$ is the binomial coefficient}, $2^{l^\mF}$ is the minimum Hamming distance of a pure polar code with the frozen set $\mF$, and {the vector $\tau_{v}$ is the sequence consisting of all bit-channel indices of weight $v$ arranged in ascending order of their reliabilities}.
That is,  $\tau_v$ is the subsequence of the reliability sequence $r$ consisting of $r_i$ with $\wt(r_i)=v$, $0\leq i < N$. 
{Note that} the reliability-based frozen set $\mR_S$ of cardinality $N-K-S$ 
\begin{equation}
\mR_S\triangleq\{r_i \,|\, i \in [N-K-S]\}
\label{eq:RS}
\end{equation}
{satisfies} ${\mR_S=}\{ \tau_{v,q} \, |\, q\in [\alpha^{\mR_S}_v], 0\leq v\leq n\}$. 
Let 
\begin{equation}
\ell\triangleq l^{\mR_0}.
\label{eq:ell}
\end{equation}
Since the bit-channels 
\begin{align}
{\mathcal{H}_{\mathrm{inf}}}&{\triangleq \{ i\in [N]\;|\;\ell+2\leq \wt(i)\leq n\}\setminus \mR_0}\label{eq:permInf}\\ &{=}\{\tau_{v,q} \,|\,  q\in [c_v]\setminus[\alpha^{\mR_0}_v],\ell+2\leq v\leq n\}\nonumber
\end{align}
have high-weight indices and high reliabilities, they are unlikely to generate codewords producing errors under SCL decoding. 
Therefore, the following constraint suggests that these bit-channels 
are always non-frozen, while the least reliable bit-channels are always frozen. 

\begin{definition}[S-constraint]
Given an integer parameter $S$, a frozen set $\mF$ satisfies the S-constraint iff 
\begin{equation*}
\begin{cases}
\alpha^\mF_v=c_v, &0\leq v< \ell, \\
\alpha^\mF_v\geq \alpha^{\mR_S}_v, &\ell\leq v\leq n,\\ 
\forall q\geq \alpha^{\mR_0}_v\;\;\tau_{v,q}\notin\mF, &\ell+2\leq v\leq n, 
\end{cases}
\end{equation*}
{where $\ell = \min_{ i\in [N]\setminus\mR_0}\wt(i)$ as follows from Eqs. \eqref{eq:l} and \eqref{eq:ell}.} 
\label{def:S}
\end{definition}

{Note that we express the three conditions 
in Definition \ref{def:S}
in terms of $\alpha^{\mF}_v$ to serve as a basis for Definition \ref{def:B}. The first condition means that $\mF$ is a superset of the frozen set of Reed-Muller code with the minimum distance $2^\ell$. The second condition means that $\mF$ is also a superset of $\mR_S$. The third condition means that $\mF\cap\mathcal{H}_{\mathrm{inf}}=\emptyset$.}

\begin{figure}
    \centering  \includegraphics[width=0.48\textwidth]{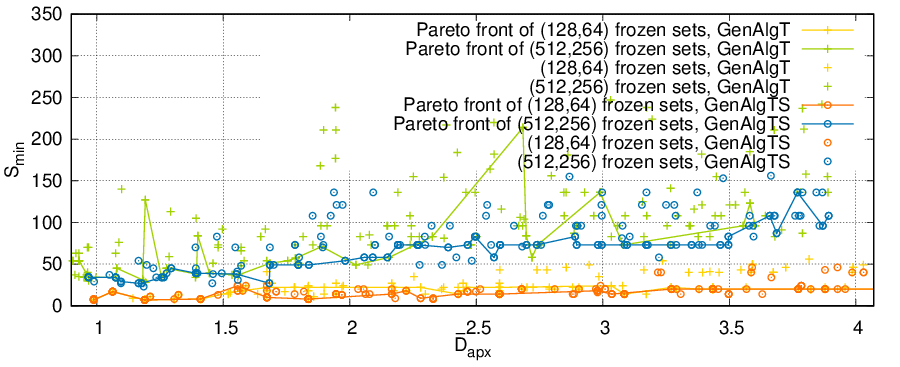}
    \caption{
    {$S_{\min}$ of} the frozen sets generated by the genetic algorithms}
    \label{fig:S_Dapprox}
\end{figure}
{Our motivation for the second condition in Definition \ref{def:S} comes from the experimental results 
 for GenAlgT shown in Fig. \ref{fig:S_Dapprox}, where ``+'' 
indicates the lowest $S$ such that 
$\mR_S\subset\mF$ for an individual frozen set $\mF$ generated by GenAlgT, i.e.,} $${S_{\min}(\mF)\triangleq\min\{ S\in [N-K] \,|\, \mR_S\subset\mF\}.}$$ 
{We focus on the best sets $\mF$ for various $\bar D_{\mathrm{apx}}$, i.e., $\mF$ belonging to the Pareto front as explained at the beginning of Section \ref{sOptimization}. These frozen sets are connected by lines in Fig. \ref{fig:S_Dapprox}. 
It can be seen that the Pareto front is characterized by a lower {$S_{\min}$} than the average {$S_{\min}$}. For $(N,K)=(512,256)$, almost all $\mF$ in the Pareto front satisfy $S_{\min}(\mF)\leq 160$, except for one outlying frozen set. Note that we are interested in the general trend and can ignore the outliers, since there exist plenty of frozen sets with almost identical $(\bar D_{\mathrm{apx}}, \widetilde{P}_{\mathrm{ML}})$, and it is very likely that one of them would have low {$S_{\min}$}. Therefore, we propose to consider only frozen sets $\mF$ being supersets of $\mathcal{\mR}_S$ with a certain $S$ in the genetic algorithm.} 

{The S-constraint integration into the genetic algorithm from Section \ref{sGenAlgT} is straightforward. The resulting algorithm is referred to as the \textit{\textbf{GenAlgTS}} and specified in Algorithm \ref{alg:GenAlgTS}. Compared to GenAlgT, GenAlgTS has an additional input argument $S$ and additional lines 5--9. At lines 5--6, we compute the set $\mathcal{H}_{\mathrm{fr}}$ of permanently frozen bits to satisfy the first and second conditions of the S-constraint. At line 7, we compute the set $\mathcal{H}_{\mathrm{inf}}$ of permanently non-frozen bits to satisfy the third condition of the S-constraint. Then we eliminate frozen sets not satisfying the S-constraint from the initial population at line 8. 
At line 9, we initialize the set $\mathcal{Z}$ of bits that might be arbitrarily frozen or non-frozen. Finally, we execute the function $Evolve$ with the set $\mathcal{Z}$. Note that after the execution of line 8, we have $\mF\cap \mathcal{H}_{\mathrm{fr}}=\mathcal{H}_{\mathrm{fr}}$ and $\mF\cap\mathcal{H}_{\mathrm{inf}}=\emptyset$ for all $\mF\in\mathbf{P}$. Function $Evolve$ preserves this property by allowing $Mutation$ and $Crossover$ to perform operations only over bit indices from $\mathcal{Z}$.} 
\begin{algorithm2e}
\small
\setstretch{0.8}
\caption{
{\small Genetic algorithm with the S-constraint}
\label{alg:GenAlgTS}}
\DontPrintSemicolon
\SetKwRepeat{Do}{do}{while}
\SetKwFunction{FInitializePopulation}{InitializePopulation}
\SetKwFunction{FRandom}{Random}
\SetKwFunction{FMutation}{Mutation}
\SetKwFunction{FCrossover}{Crossover}
{\textbf{GenAlgTS}$(N,K,T_{\mathrm{POP}},T_D,\theta,S)$} 

{\Begin{
    $\mathbf{P}\leftarrow \FInitializePopulation(N,K)$\;
    $\mathbf{P}\leftarrow \{ \mF\in\mathbf{P}\;|\; \bar D_{\mathrm{apx}}(\mF) \leq T_D\}$\;
    Compute $\mR_S$ in Eq. \eqref{eq:RS} and $\ell$ in Eq. \eqref{eq:ell}\;
    $\mathcal{H}_{\mathrm{fr}}\leftarrow \mR_S\cup \{i\in [N] \;|\; \wt(i)<\ell\}$\;
    Compute $\mathcal{H}_{\mathrm{inf}}$ in Eq. \eqref{eq:permInf}\;
    $\mathbf{P}\leftarrow \{ \mF\in\mathbf{P}\;|\;  \mathcal{H}_{\mathrm{fr}}\subset\mF, \mathcal{H}_{\mathrm{inf}}\cap\mF=\emptyset\}$\;
    $\mathcal{Z}\leftarrow [N]\setminus (\mathcal{H}_{\mathrm{fr}}\cup\mathcal{H}_{\mathrm{inf}})$\;
    $\mF^*\leftarrow\FEvolve(N,K,T_{\mathrm{POP}},T_D,\theta,\mathbf{P}, \mathcal{Z})$\;
    \Return $\mF^*$\; 
}}
\end{algorithm2e}

Fig. \ref{fig:SConstrFrSet} illustrates the S-constraint on the frozen set structure. The blue colour indicates  the {set $\mathcal{H}_{\mathrm{fr}}$ of} permanently frozen bits, whose {cardinality} is $\mathcal{N}^{
\mathrm{fr}
}_S\triangleq(\sum_{v=0}^{\ell-1} c_v) + (\sum_{v=\ell}^{n} \alpha^{\mR_S}_v)$. The red colour indicates the {set $\mathcal{H}_{\mathrm{inf}}$ of} permanently non-frozen bits, whose {cardinality} is $\mathcal{N}^{
\mathrm{inf}
}\triangleq\sum_{v=\ell+2}^{n} (c_v-\alpha^{\mR_0}_v)$. 
\begin{figure}
\centering
\includegraphics[width=0.4\textwidth]{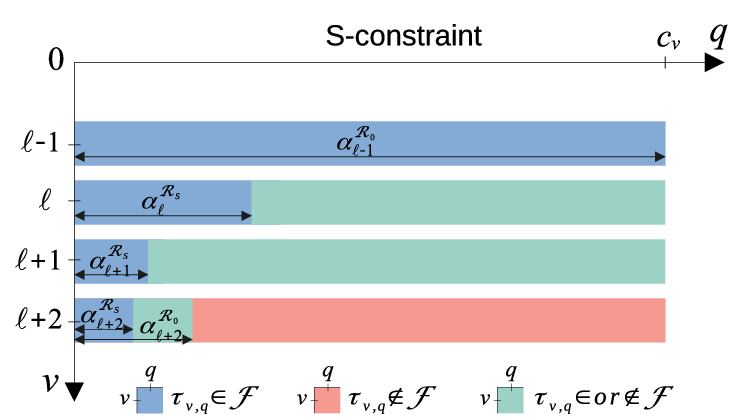}
    \caption{S-constrained frozen set structure}
    \label{fig:SConstrFrSet}
\end{figure}
These $\mathcal{N}^{\mathrm{fr}}_S+\mathcal{N}^{\mathrm{inf}}$ bits {are} eliminated from consideration during the frozen set optimization {in GenAlgTS}. Thus, the frozen set optimization under the S-constraint reduces to the optimization over the green area {corresponding to the set $\mathcal{Z}$} consisting of ${|\mathcal{Z}|=}N-\mathcal{N}^{\mathrm{fr}}_S-\mathcal{N}^{\mathrm{inf}}$ bits. 
The number of $(N,K)$ S-constrained frozen sets is $\Omega_S\triangleq\binom{N-\mathcal{N}^{\mathrm{fr}}_S-\mathcal{N}^{\mathrm{inf}}}{K-\mathcal{N}^{\mathrm{inf}}}$. Obviously, 
$\Omega_S$ increases with $S$ from $\min_S(\Omega_S)=\Omega_{0}{=1}$ to $\max_S(\Omega_S)=\Omega_{N-K}$. 
Since {$\Omega_S$} rapidly grows with $S$, it is desirable to identify the lowest $S$ that preserves the best frozen sets. 
{For example, we can calculate $\Omega_S$ for $S=160$.} 
For $(N,K)=(512,256)$, we have $\ell=4$, $\mathcal{N}^{\mathrm{inf}}=129$, $\mathcal{N}^{\mathrm{fr}}_S=138$, and therefore $\Omega_S=\binom{245}{127}$. 
For $(N,K)=(128,64)$, we have 
$\ell=3$, $\mathcal{N}^{\mathrm{fr}}_S=29$, $\mathcal{N}^{\mathrm{inf}}=29$, and consequently $\Omega_S=\binom{70}{35}$. In both cases, $\Omega_S$ is much lower than $\binom{N}{K}$. 
This complexity reduction is achieved without reducing the Pareto front quality, as shown in Section \ref{sGeneratedFrSets}. 
{Furthermore, a lower value of $S$ would suffice as can be seen from the GenAlgTS Pareto front in Fig. \ref{fig:S_Dapprox}.} 

{For the general case of $(N,K)$, a proper value of $S$ can be found by running GenAlgTS with $S$ equal to $0$ and then increasing $S$ as long as $\widetilde{P}_{\mathrm{ML}}$ is decreasing for a given $\bar D_{\mathrm{apx}}
$. Since $S$ generally increases with $\bar D_{\mathrm{apx}}
$, it suffices for find a proper value of $S$ for the largest expected $\bar D_{\mathrm{apx}}
$. Note that $S$ is lower bounded by $0$ and upper bounded by $N-K$. $S$ may be increased with a fixed step size, e.g., $40$, or a variable step size.}   

\subsubsection{{Complexity} Reduction. B-Constraint and GenAlgTB}
\label{sGenAlgTB}

The gradual increase of {$S_{\min}$ in Fig. \ref{fig:S_Dapprox}} is due to the growing 
discrepancy between our frozen sets $\mF$ and the reliability-based frozen set ${\mR_0}$. This growing discrepancy can be characterized not only by {$S_{\min}$} but also by the number $\Delta$ of frozen bit-channels with the highest indices of weight $\ell$ 
$${\Delta(\mF)\triangleq c_{\ell}-1-\max \{q\in [c_{\ell}] \,|\, \tau_{\ell,q}\notin\mF\}.}$$ {In \cite[Corollary 1]{milos2024Reinf}, we have proven that $\Delta$ must be equal or higher than $n-\ell+1$ for 
 precoded polar codes with the minimum distance $\geq 1.5\cdot 2^{\ell}$ being subcodes of Reed-Muller codes with the minimum distance $2^{\ell}$.} {Note that Definition \ref{def:S} ensures that the minimum distance of precoded polar codes with the S-constrained frozen sets is lower bounded\footnote{{The first 
 S-constraint requirement, i.e. $\alpha^\mF_v=c_v$ for $0\leq v<\ell$ 
 in Definition \ref{def:S}, 
 ensures that all information bit indices have weights $\geq\ell$. Precoded polar codes with such information bit indices have the minimum distance $\geq 2^{\ell}$ \cite[Section III]{Rowshan2022Impr}.}} by $2^{\ell}$.}  
\begin{figure}
    \centering  \includegraphics[width=0.48\textwidth]{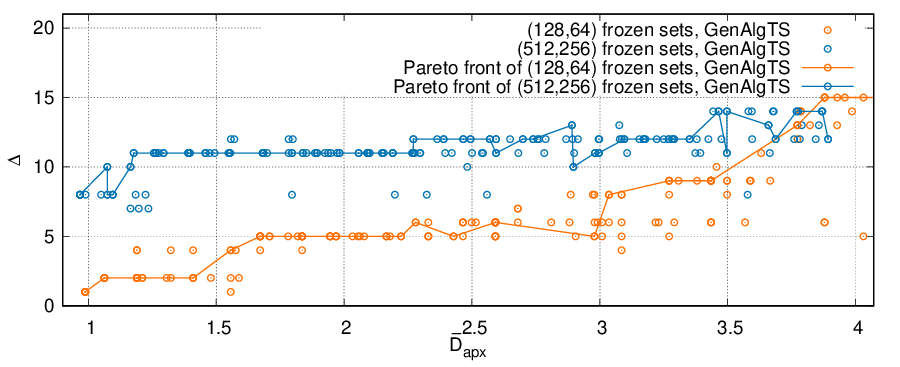}
    \caption{$\Delta$ of the frozen sets generated by the genetic algorithm}
    \label{fig:Delta}
\end{figure}
{Fig. \ref{fig:Delta} shows that $\Delta$ of the S-constrained frozen sets increases with $\bar D_{\mathrm{apx}}$.} {We focus on the behaviour of $\Delta$ for $\mF$ belonging
to the Pareto front, as in the case of $S_{\min}$ in Fig. \ref{fig:S_Dapprox}.} 

Motivated by the results on $\Delta$, we propose to incorporate similar characteristics for the two lowest information index weights into the frozen set structure. 
Specifically, we represent $\Delta$ of a given frozen set $\mF$ as $\Delta{(\mF)}=c_\ell-1-\chi^\mF_\ell$ 
and propose to optimize $\chi^\mF_{l^\mF}$ and $\chi^\mF_{l^\mF+1}$, where 
\begin{equation}\chi^\mF_v\triangleq \max_{q\in [c_v]}\{ z_{v,q}\notin\mF\},\quad l^\mF\leq v \leq n,
\label{eq:chi}
\end{equation}
and $z_v$ is the subsequence of $(0,1,\dots,N-1)$ consisting of all elements of weight $v$. Obviously, the length of $z_v$ is $c_v$. The value of $\chi^\mF_{l^\mF}$ and $\chi^\mF_{l^\mF+1}$ are expected to decrease with increasing $\bar D_{\mathrm{apx}}$ due to their connection with $\Delta$. Besides, we propose to assume that highly reliable bit-channels having $v$-weight indices less than $\chi^\mF_{v}$ are non-frozen, $v\in\{ l^\mF,l^\mF+1\}$. This assumption reduces the number of bit-channels allowed to be arbitrarily frozen or non-frozen. The resulting frozen set structure is 
formalized as the B-constraint, where the flexibility of frozen set $\mF$ 
is controlled by integers $B_{l^\mF}$ and $B_{l^\mF+1}$. 

\begin{definition}[B-constraint]
\label{def:B}
Given an integer vector $B$,  
a frozen set $\mF$ satisfies the B-constraint iff  
\begin{equation*}
\begin{cases}
\beta^\mF_v< \alpha^\mF_v+B_v, & l^\mF\leq v\leq l^\mF+1,\\ 
\beta^\mF_v=\alpha^\mF_v-1, & l^\mF+2\leq v\leq n, 
\end{cases}
\end{equation*}
where 
\begin{equation*}
\beta^\mF_v\triangleq \max\{ q\in [c_v] \,|\, \tau_{v,q}\in\mF, \tau_{v,q}< \chi^\mF_v\}, \quad l^\mF\leq v\leq n.
\end{equation*}
\end{definition}
\begin{figure}
\centering
\includegraphics[width=0.4\textwidth]{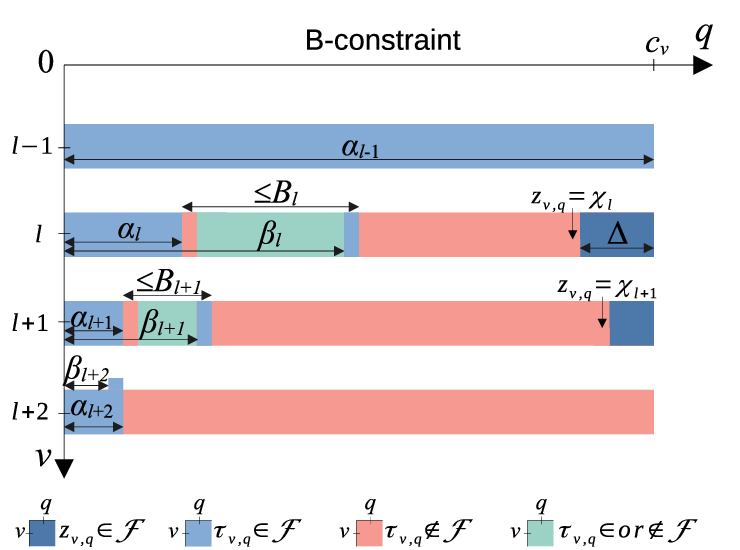}
    \caption{B-constrained frozen set structure}
    \label{fig:BConstrFrSet}
\end{figure}
The B-constraint on the frozen set structure is illustrated by Fig. \ref{fig:BConstrFrSet}, where the upper index ``$\mF$'' is omitted for simplicity. Note that the dark and light blue frozen bits are arranged in ascending order of their indices and reliabilities, respectively. That is, the $v$-weight dark and light blue frozen bits are arranged as in the sequences $z_v$ and $\tau_v$, respectively. The rest of the notation is as in Fig. \ref{fig:SConstrFrSet}.


\begin{remark}
The B-constrained frozen set structure generalizes our frozen set structure \cite{milos2024Reinf}. Specifically, the triplet-tuned frozen sets from \cite{milos2024Reinf} can be represented as special cases of 
the B-constrained frozen sets with  $B_{l^\mF}=B_{l^\mF+1}=0$ and $\chi^\mF_{v}=z_{v,c_v-1}$ for $v=l^\mF+1$.
\end{remark}

The number of B-constrained frozen sets $\mF$ with fixed $l^\mF$ is upper bounded by
$$ \left[\prod_{v=l^\mF}^{l^\mF+1}(c_v-B_v)2^{B_v}\right]\left[\prod_{v=l^\mF+2}^{n}c_v\right] c_{l^\mF}.$$ 
 By imposing limits on $\alpha^\mF_v$ as in Definition \ref{def:S}, we obtain the following upper bound on the number of $\mF$ satisfying both B-constraint and S-constraint with $l^\mF=\ell$
$$ \left[\prod_{v=\ell}^{\ell+1}(c_v-\alpha^{\mR_S}_v-B_v)2^{B_v}\right]\left[\prod_{v=\ell+2}^{n}(\alpha^{\mR_0}_v-\alpha^{\mR_S}_v+1)\right]c_\ell.$$
This number is minimized when $B_v=0$ and maximized when $B_v=c_v-\alpha^{\mR_S}_v-1$. 



\begin{algorithm2e}
\small
\setstretch{0.8}
\caption{\small Genetic algorithm with the B-constraint
\label{alg:GenAlgTB}}
\DontPrintSemicolon
\SetKwRepeat{Do}{do}{while}
\SetKwFor{Repeat}{repeat}{}{end}
\SetKwFunction{FInitializePopulation}{InitializePopulation}
\SetKwFunction{FPrunePopulation}{PrunePopulation}
\SetKwFunction{FFlexibleElements}{FlexibleElements}
\SetKwFunction{FRandom}{Random}
\SetKwFunction{FMutation}{Mutation}
\SetKwFunction{FBest}{Best}
{\textbf{GenAlgTB}$(N,K,T_{\mathrm{POP}},T_D,\theta,B,X)$} 

{\Begin{
$\mathbf{P}\leftarrow \FInitializePopulation(N,K)$\;
    $\mathbf{P}\leftarrow \{ \mF\in\mathbf{P}\;|\; \bar D_{\mathrm{apx}}(\mF) \leq T_D\}$\;
    $\mathbf{P}\leftarrow \{ \mF\in\mathbf{P}\;|\;\mF$ satisfies both conditions in  Def. \ref{def:B}$\}$\;
$\mathbf{P}\leftarrow \FPrunePopulation(\mathbf{P},T_{\mathrm{POP}})$\;
\For{$t=0,\dots,\theta-1$}{
    $\mathbf{P}'\leftarrow \mathbf{P}$\; 
    \For{$\mF\in\mathbf{P}'$}{  
    Compute the set $\mathcal{Z}$ for $\mF$ in Eq. \eqref{eq:flexZ}\; 
    \Repeat{$\mathrm{3\; times}$}{ 
    $\mathbf{P}\leftarrow \mathbf{P}\cup\{\FMutation(\mF,T_D,\mathcal{Z})\}$\;
    }} 
    $\mathbf{P}\leftarrow \FPrunePopulation(\mathbf{P},T_{\mathrm{POP}})$\;
    \textbf{if} $\widetilde{P}_{\mathrm{ML}}(\FBest(\mathbf{P}))$ has reduced \textbf{then} $t\leftarrow 0$\;
    }
     \Return $\FBest(\mathbf{P})$
}}

\end{algorithm2e}

Observe that in contrast with the S-constraint, the B-constraint cannot be integrated into GenAlgT by simply eliminating a subset of bit-channels from consideration. {The proposed B-constrained genetic algorithm, referred to as \textit{\textbf{GenAlgTB}}, is shown in Algorithm \ref{alg:GenAlgTB}. Lines 3--4 are as in GenAlgT and GenAlgTS. At line 5, we eliminate frozen sets not satisfying the B-constraint from the initial population $\mathbf{P}$. 
It can be seen that lines 6--18 are similar to the lines of the function $Evolve$, specified in Algorithm \ref{alg:GenAlgT} and used by GenAlgT and GenAlgTS. The difference consists in replacing the call of $ExtendPopulation$ in $Evolve$ by lines 8--14 in GenAlgTB. Lines 8--14 implement the population extension procedure with the B-constraint. Note that the B-constraint is incompatible with the function $ExtendPopulation$, because $ExtendPopulation$ uses the same $\mathcal{Z}$, which is the set of bits that might be arbitrarily frozen or non-frozen, for all frozen sets $\mF$ in the population $\mathbf{P}$. In $ExtendPopulation$, the population is extended by applying $Crossover$ and $Mutation$ to $\mF\in\mathbf{P}$. $Crossover$ does not preserve the B-constrained frozen sets, and therefore, we replace the call of $Crossover$ by two additional calls of $Mutation$ at lines 11-13 of GenAlgTB. The number of additional $Mutation$ calls is two for the following reasons. Since 
$T_{\mathrm{POP}}=5$ as in \cite{Elkelesh2019}, the crossover generates $T_{\mathrm{POP}}(T_{\mathrm{POP}}-1)/2=10$ frozen sets. To preserve the maximum population size $20$, we apply two additional mutations to each frozen set $\mF$ from the truncated population $\mathbf{P}'$ 
instead of the crossover since this generates $2T_{\mathrm{POP}}=10$ frozen sets. Thus, for each $\mF\in\mathbf{P}'$, we produce three frozen sets using $Mutation$ at lines 10--13 and add them to the extended population $\mathbf{P}$.} 


{It remains to specify the computation of $\mathcal{Z}$ for a given frozen set $\mF$ at line 10 of GenAlgTB. We first explain a general idea and then provide an equation for $\mathcal{Z}$.} To preserve the B-constrained structure of a frozen set $\mF$, we allow the mutation operation to perform only the following actions: (i) increment $\chi^{\mF}_v$, $l^\mF\leq v\leq l^\mF+1$, (ii) decrement $\beta^{\mF}_v$, $l^\mF+2\leq v\leq n$, or (iii) modify bits within the flexible region of size $B_{l^\mF}+B_{l^\mF+1}$. Action (i) is due to the initial population 
consisting of the reliability-based frozen sets and frozen sets interpolating between the Reed-Muller and reliability-based frozen sets that have $\chi^{\mF}_v=0$. Thus, $\chi^{\mF}_v$ gradually increases with the increasing number of genetic algorithm iterations. Action (ii) is since $\alpha^{\mF}_v$ and $\beta^{\mF}_v$ are expected to decrease with the increasing $\chi^{\mF}_v$. Action (iii) makes use of the flexibility allowed by $B$. 
These actions are implemented as a random swap of a frozen bit and a non-frozen bit from the set {$\mathcal{Z}$ equal to} 
\begin{align}
\{ \tau_{v,\min(\beta^\mF_v,\beta^{\mR_0}_v)+X_v-i},\chi^\mF_v\,|\, i\in [B_v]\}_{v=l^\mF}^{l^\mF+1} 
\cup \{ \tau_{v,\beta^\mF_v}\}_{v=l^\mF+2}^n,
\label{eq:flexZ}
\end{align}
where $X_v$ is an integer parameter, and $\tau_{v,q}$ with $q\notin[c_v]$ are skipped. 

It is easy to see {from Eq. \eqref{eq:flexZ}} that the cardinality of {the set $\mathcal{Z}$} is at most $\mathcal{N}^{\mathrm{flex}}\triangleq n-l^\mF+1+B_{l^\mF}+B_{l^\mF+1}$, which is significantly lower than $N-\mathcal{N}^{\mathrm{fr}}_S-\mathcal{N}^{\mathrm{inf}}$ in GenAlgTS and $N$ in {GenAlgT and} \cite{Elkelesh2019}. Thus, only a small portion of bits are allowed to mutate at each iteration of the genetic algorithm. Although the {number of possible frozen sets} can be minimized by using $B_{l^{\mF}}=B_{l^{\mF}+1}=0$, such a choice leads to a rigid frozen set structure and may eliminate many good frozen sets from consideration. Therefore, the values of $B_{l^{\mF}}$ and $B_{l^{\mF}+1}$ are selected to balance the frozen set flexibility and its design complexity. We set $X_{l^{\mF}}=8$, $X_{l^{\mF}+1}=6$, $B_{l^{\mF}}=37$, and $B_{l^{\mF}+1}=8$ for the code parameters $(128,64)$ and $B_{l^{\mF}+1}=23$ for $(512,256)$.  
This defines the number of bits allowed to mutate $\mathcal{N}^{\mathrm{flex}}=50$ for the code parameters $(128,64)$ and $\mathcal{N}^{\mathrm{flex}}=66$ for $(512,256)$. {For the code parameters $(512,128)$, the same settings can be used as for $(512,256)$. However, the optimization complexity can be further reduced by setting $X_{l^{\mF}}=6$, $X_{l^{\mF}+1}=0$, $B_{l^{\mF}}=24$, and $B_{l^{\mF}+1}=5$ for $(512,128)$ without performance degradation, leading $\mathcal{N}^{\mathrm{flex}}=34$.} 
\section{Numerical Results}
\label{sNumerical}

In this section, we evaluate the proposed frozen set design method and provide a comparison with the state-of-the-art for the AWGN channel with BPSK modulation. 
\subsection{Frozen Set Design} 
\label{sGeneratedFrSets}


\begin{figure}
    \centering
    \includegraphics[width=0.47\textwidth]{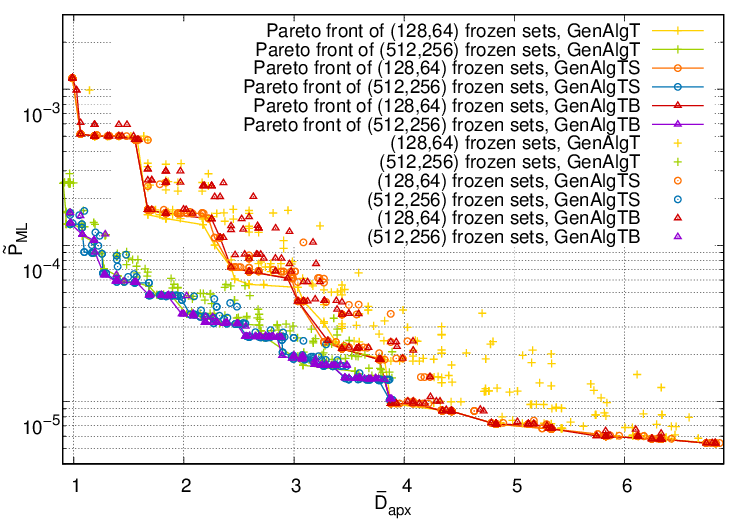}
    \caption{Frozen sets generated by the genetic algorithms}
    \label{fig:Pareto_Dapprox}
\end{figure}

Fig. \ref{fig:Pareto_Dapprox} characterizes the frozen sets generated by the proposed GenAlgT, GenAlgTS and GenAlgTB for $(N,K)\in\{ (128,64), (512,256) \}$, $T_D\in\{ 1.0, 1.1, 1.2, \dots\}$ and $\rho=5${, where $T_D$ is the upper limit for $\bar D_{\mathrm{apx}}$ and $\rho$ is the number of genetic algorithm runs, as defined in Section \ref{sGenAlgT}.}  The pairs $(\bar D_{\mathrm{apx}},\widetilde{P}_{\mathrm{ML}})$ found by GenAlgT, GenAlgTS and GenAlgTB are marked as ``+'', ``$\circ$'' and ``\begin{scriptsize}$\triangle$\end{scriptsize}'', respectively, where $\widetilde{P}_{\mathrm{ML}}$ is computed using {the ensemble-averaged weight distribution\footnote{{The weight distribution is averaged over the ensemble of precoded polar codes with a given frozen set and all possible frozen bit expressions.}} 
with} the low-complexity union bound for intermediate iterations and the tight TSB bound for the final output at $E_b/N_0=3.5$ dB for the code parameters $(128, 64)$ and 
 $E_b/N_0=2.0$ dB for $(512,256)$. 
It can be seen that GenAlgT, GenAlgTS and GenAlgTB provide similar Pareto fronts, indicating that the {complexity} reduction of GenAlgTS does not deteriorate the frozen set performance. 
Moreover, the outputs of GenAlgTS and GenAlgTB 
are concentrated closer to the Pareto front than that of GenAlgT. This is because GenAlgTS and GenAlgTB have fewer local optima than GenAlgT 
due to the reduced {number of possible solutions}. As a result, GenAlgTS and GenAlgTB need a lower $\rho$ to reach saturation than GenAlgT, where the saturation is achieved if an increase in $\rho$ does not provide any reduction of $\min_{i\in [\rho]} \widetilde{P}_{\mathrm{ML},i}$, where $\widetilde{P}_{\mathrm{ML},i}$ is the $i$-th run output of GenAlgT/GenAlgTS/GenAlgTB.
That is why the Pareto front of GenAlgTS/GenAlgTB is slightly better on average than that of GenAlgT in the case of parameters $(512,256)$. In the case of $(128,64)$, 
the Pareto front of GenAlgT is slightly better on average than that of GenAlgTS/GenAlgTB, 
since for short-length codes, the {number of possible solutions in} GenAlgT is small enough to find near-optimal solutions. 
Note that the computational complexities  
of GenAlgT and GenAlgTS
grow rapidly with the code length $N$, 
while the complexity of GenAlgTB grows slowly with $N$, as follows from the description in Section \ref{sOptimization}. 
Besides, it can be seen from Fig. \ref{fig:Pareto_Dapprox} that the Pareto front of $(128,64)$ frozen sets has a more stepwise character than that of $(512,256)$ frozen sets. This implies that the Pareto front becomes smoother with increasing code length $N$. 



The computational complexity of genetic algorithms is often characterized by the number of iterations. 
\begin{figure}
    \centering
    \includegraphics[width=0.47\textwidth]{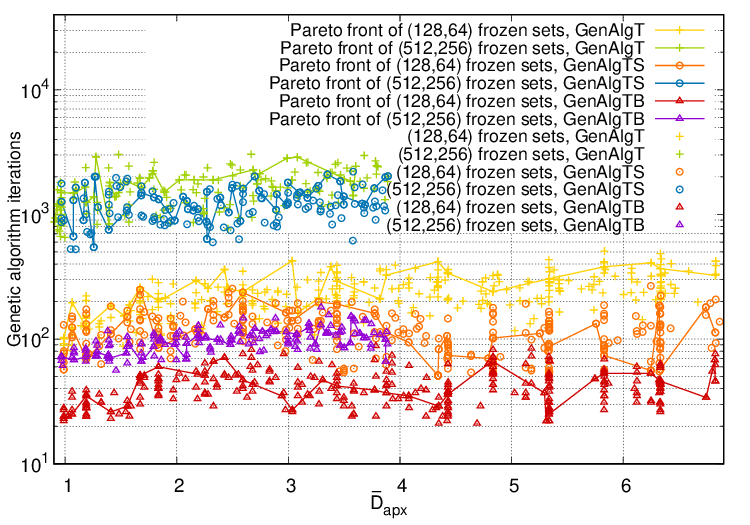}
    \caption{The number of iterations in GenAlgT, GenAlgTS and GenAlgTB}
    \label{fig:IterTotal_Dapprox}
\end{figure}
According to Fig. \ref{fig:IterTotal_Dapprox}, GenAlgT, GenAlgTS and GenAlgTB perform $239$, $117$ and $42$ iterations on average for $(128,64)$, respectively. GenAlgT, GenAlgTS and GenAlgTB perform $1617$, $1220$ and $98$ iterations on average for $(512,256)$, respectively. Note that GenAlgT and GenAlgTS terminate if no improvement has been observed for the last $50$ and $200$ iterations for the parameters $(128,64)$ and $(512,256)$, respectively. In GenAlgTB, the numbers of such last iterations are $20$ and $30$ for $(128,64)$ and $(512,256)$, respectively. So, GenAlgT, GenAlgTS and GenAlgTB found the resulting frozen sets in $189$, $67$ and $22$ iterations on average for $(128,64)$, respectively. GenAlgT, GenAlgTS and GenAlgTB found the resulting frozen sets in $1417$, $1020$ and $68$ iterations on average for $(512,256)$, respectively. Thus, GenAlgTB requires much less iterations than GenAlgT.   

The execution of GenAlgTB required $0.2$ and $10$ seconds on average for the parameters $(128,64)$ and $(512,256)$, respectively, whereas the resulting frozen sets were found after $0.1$ and $6$ seconds on average for $(128,64)$ and $(512,256)$, respectively. The implementation is non-parallel and executed on a computer with i7 3.2GHz processor. Note that the complexity is independent of the design $E_b/N_0$, since the code performance is evaluated via theoretical bounds.

\begin{figure}[htp]
    \centering
     \begin{subfigure}{0.47\textwidth}
         \centering
         \includegraphics[width=\textwidth]{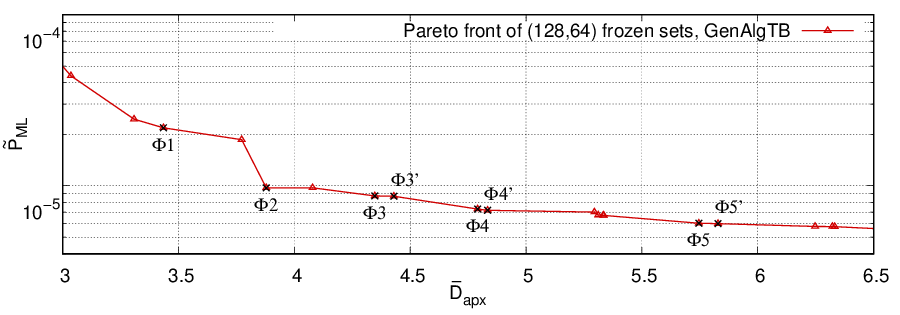}
     \end{subfigure}
     \hfill
     \begin{subfigure}{0.47\textwidth}
         \centering
         \includegraphics[width=\textwidth]{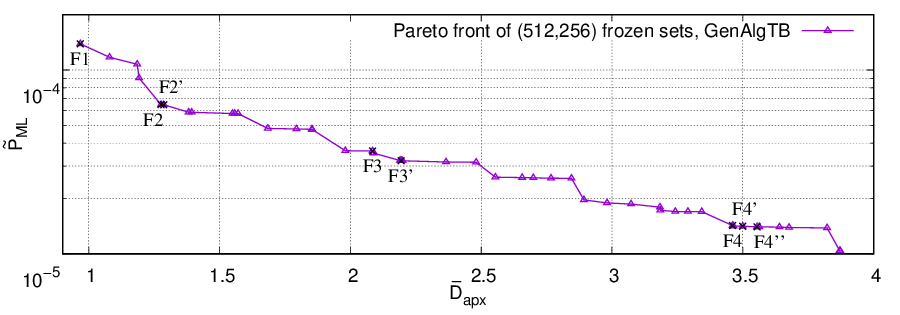}
     \end{subfigure}
     \hfill
     \begin{subfigure}{0.47\textwidth}
         \centering
         \includegraphics[width=\textwidth]{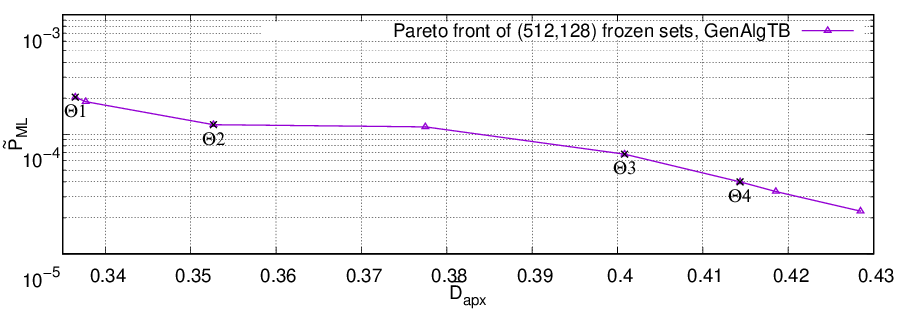}
         \end{subfigure}
    \caption{{Frozen sets at the Pareto front of GenAlgTB are indicated by $\triangle$. Those of them evaluated in Section \ref{sFER} are indicated by $\times$ and the corresponding labels. 
    }}
    \label{fig:FrSetLabels}
\end{figure}

    {Since GenAlgTB has lower complexity than GenAlgT and GenAlgTS, we further consider only frozen sets generated by GenAlgTB. For clarity, the Pareto fronts of GenAlgTB are shown separately for $(128,64)$, $(512,256)$ and $(512,128)$ frozen sets in Fig. \ref{fig:FrSetLabels}. Note that the Pareto fronts of GenAlgTB in Fig. \ref{fig:Pareto_Dapprox} are the same as in Fig. \ref{fig:FrSetLabels}. 
} 

\subsection{Performance of Precoded Polar Codes}
\label{sFER}

In Section \ref{sGeneratedFrSets}, we evaluated the proposed frozen set design. {The frozen sets from the Pareto front of GenAlgTB, shown in Fig. \ref{fig:FrSetLabels},} are further integrated with the frozen bit expressions from Section \ref{sFrBitExpressions} to yield precoded polar codes. 
\begin{figure}
    \centering    \includegraphics[width=0.47\textwidth]{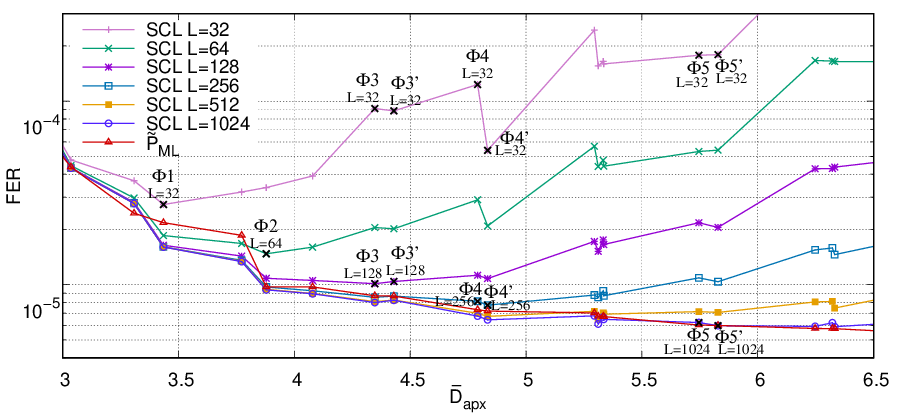} 
    \caption{{The FER performance of the $(128,64)$ Pareto front from Fig. \ref{fig:FrSetLabels} under SCL decoding} 
    }  \label{fig:CodeSelection_n128_k64}
\end{figure}
{
Fig. \ref{fig:CodeSelection_n128_k64} illustrates the FER performance of the precoded polar codes with all $(128,64)$ frozen sets from Fig. \ref{fig:FrSetLabels} under the SCL decoder \cite{tal2015list} with $L\in\{32,64,128,256,512,1024\}$ at $E_b/N_0=3.5$ dB.
It can be seen that for each $L$, the SCL FER generally decreases until reaching
the global minimum 
and then generally increases. The SCL FER  minimum 
corresponds to the frozen set(s) chosen for the list size $L$, e.g., $\Phi 1$ for $L=32$,  $\Phi 2$ for $L=64$, $\Phi 3$ and $\Phi 3'$ for $L=128$, $\Phi 4$ and $\Phi 4'$ for $L\in\{256,512\}$, and $\Phi 5$ and $\Phi 5'$ for $L=1024$. Fig. \ref{fig:CodeSelection_n128_k64} also shows the  $\widetilde{P}_{\mathrm{ML}}$ curve from Fig. \ref{fig:FrSetLabels}, which monotonically decreases in accordance with the Pareto front definition. Since the SCL decoding approaches the ML decoding when $L$ is large enough for a given code, the SCL FER decreases almost monotonically when it is close to the $\widetilde{P}_{\mathrm{ML}}$ curve. 
It can be seen that the frozen sets $\Phi3$ and $\Phi3'$ perform similarly for all considered $L$. The same property holds for $\Phi5$ and $\Phi5'$. Although $\Phi4$ and $\Phi4'$ perform similarly under SCL with $L\in\{128,256,512,1024\}$, the performance of $\Phi4$ for $L=32$ is substantially better than that of $\Phi4'$, $\Phi3$, $\Phi3'$, $\Phi5$ and $\Phi5'$. The outstanding performance of $\Phi4'$ under SCL with low $L$ can be explained by its low SC decoding error probability. Specifically, $\widetilde{P}_{\mathrm{SC}}(\Phi4')=0.069$ 
is noticeably lower than
$\widetilde{P}_{\mathrm{SC}}(\Phi3)=0.106$, 
$\widetilde{P}_{\mathrm{SC}}(\Phi3')=0.266$, 
$\widetilde{P}_{\mathrm{SC}}(\Phi4)=0.113$, 
$\widetilde{P}_{\mathrm{SC}}(\Phi5)=0.114$ 
and $\widetilde{P}_{\mathrm{SC}}(\Phi5')=0.272$ at $E_b/N_0=3.5$ dB, 
where $\widetilde{P}_{\mathrm{SC}}$ is the analytical approximation
of the SC decoding error probability \cite[Eq. (3)]{PeSC2014}. 
Thus, if we would need to select a single frozen set for $L\in\{32,64,128,256,512,1024\}$, then we would choose $\Phi4'$ based on its 
 $\widetilde{P}_{\mathrm{ML}}$, $\bar D_{\mathrm{apx}}$ and $\widetilde{P}_{\mathrm{SC}}$. 
 However, this paper focuses on designing frozen sets for each $L$ individually.}   
 
 {Note that we have shown the FER performance of $(128,64)$ codes with all frozen sets from the Pareto front of GenAlgTB under SCL with all $L\in\{32,64,128,256,512,1024\}$ in Fig. \ref{fig:CodeSelection_n128_k64} only to illustrate the behaviour of the Pareto front depending on 
the decoding list size $L$. 
A proper frozen set for a given target $L$ can be found
by running decoding simulations only for a small number of frozen sets from the Pareto front. 
Specifically, the knowledge that the SCL FER decreases almost monotonically before reaching the global minimum and is almost convex in the vicinity of the global minimum 
can be employed to find the SCL FER global minimum with a small number of considered frozen sets. This can be done by first finding an approximate global minimum by considering frozen sets from the Pareto front with a large step in $\widetilde{P}_{\mathrm{ML}}$ (alternatively, a large step in $\bar D_{\mathrm{apx}}$) and then reducing the step size to find the minimum accurately. To reduce the computational complexity of the search for an approximate minimum, we recommend using rough FER estimates obtained by running the decoder only until a few dozens of errors. The subsequent search in the vicinity of the approximate minimum requires accurate FER estimates, for example, we run decoding simulations until 1000 errors for the code length $128$ and until 400 errors for the code length $512$.
If a frozen set for a certain $L$ is already found, then the frozen sets having higher $\widetilde{P}_{\mathrm{ML}}$ (alternatively, lower $\bar D_{\mathrm{apx}}$)  
can be eliminated from the consideration when searching for frozen sets for a larger target $L$, and similarly, the frozen sets having lower  $\widetilde{P}_{\mathrm{ML}}$ (alternatively, higher $\bar D_{\mathrm{apx}}$) can be eliminated in the case of a smaller target $L$.  
Note that the candidate frozen sets can be generated one by one using GenAlgTB with the input parameter $T_D$ slightly exceeding the desirable $\bar D_{\mathrm{apx}}$, instead of generating the whole Pareto front beforehand. 
}

In what follows, we compare the FER performances of the proposed codes and the state-of-the-art codes. 
The codes are labelled as follows:
\begin{small}
\begin{singlespace}
\begin{itemize}
\item {\textbf{Proposed $\pmb{\Phi}$\_}, \textbf{Proposed F\_} and \textbf{Proposed $\pmb{\Theta}$\_}}\footnote{{The proposed precoded polar codes, i.e., the proposed frozen sets and generator matrices, 
are available at {https://sites.google.com/site/veradmiloslavskaya/specifications-of-error-correcting-codes}.}} -- precoded polar codes with the proposed frozen sets from {Fig. \ref{fig:FrSetLabels}} 
and frozen bit expressions from Section \ref{sFrBitExpressions}. 
\item \textbf{5G polar CRC-11} -- 5G polar codes with CRC-11 \cite{polarNR2018}.
\item \textbf{eBCH subcode d=\_} -- eBCH polar subcodes \cite{trifonov2016subcodes} with the minimum distance \textbf{d}.
\item \textbf{Code-0}, \textbf{Code-1} and \textbf{Code-2} -- $(128,64)$ 
code from \cite[Fig. 2]{Coskun2022InfTheor}, $(512,256)$ Code-1 and Code-2 from \cite[Figs. 4 and 6]{Coskun2022InfTheor}, respectively. 
\item \textbf{PAC-RM} -- $(128,64)$ PAC code with the Reed-Muller frozen set \cite{Arkan2019FromSD}. 
\item \textbf{Systematic PAC} -- $(128,64)$ systematic PAC code 
generated by the genetic algorithm in which the minimum-weight codewords are computed at each iteration \cite[Fig. 5a]{Tonnellier2021SystPAC}. 
\item \textbf{RecursCode} -- $(128,64)$ precoded polar code, which is obtained by recursively optimizing the weight distribution of a subcode of the Plotkin sum of shorter codes \cite[Fig. 4d]{milos2021recursive}.
\item {\textbf{GNN IMP} -- $(128,64)$ code generated by the heterogeneous graph-neural network (GNN) based iterative message-passing (IMP) algorithm \cite[Fig. 10]{Liao2023GNN}.} 
\item \textbf{Number-polar s=(\_,\_,\_) GA} -- $(512,256)$ precoded polar codes specified by the triplets \textbf{s} and the Gaussian approximation-based reliability sequence \cite[Fig. 7a]{milos2024Reinf}.
\item {\textbf{PAC P=10} and \textbf{PAC P=12} -- PAC codes, whose frozen sets are generated by a greedy approach using the ML performance bound and integer parameter \textbf{P}
\cite[Figs. 3 and 5]{Chiu2023DesignP}.}
\item {\textbf{PAC+} -- $(512,128)$ PAC code, whose frozen set design emerged from the theoretical analysis of the minimum weight codewords 
\cite[Fig. 2]{Rowshan2022Impr}.}
\item \textbf{Normal approximation bound} 
\cite{Erseghe2016}. 
\end{itemize}
\end{singlespace}
\end{small}


We use the well-known SCL decoder \cite{tal2015list} for moderate list sizes $L$. When $L$ is large, we employ the sequential (SQ) decoder \cite{miloslavskaya2014sequential, Trifonov2018ASF}. Note that SQ is a variation of SCL with a similar FER performance and time complexity approaching $O(N \log(N))$ in the high-SNR region \cite{Trifonov2018ASF}, while the time complexity of SCL scales as $O(L N \log(N))$.


\begin{figure}
    \centering
    \includegraphics[width=0.47\textwidth]{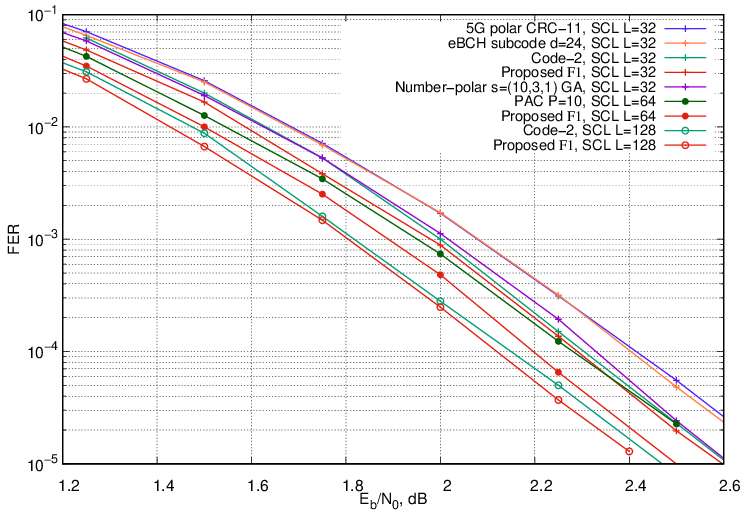} 
    \caption{{The performance comparison of  $(512,256)$ codes under SCL decoding with $L\in\{ 32,64,128\}$} 
    }
    \label{fig:FER_n512_k256_uptoL128}
\end{figure}
\begin{figure}
    \centering
    \includegraphics[width=0.47\textwidth]{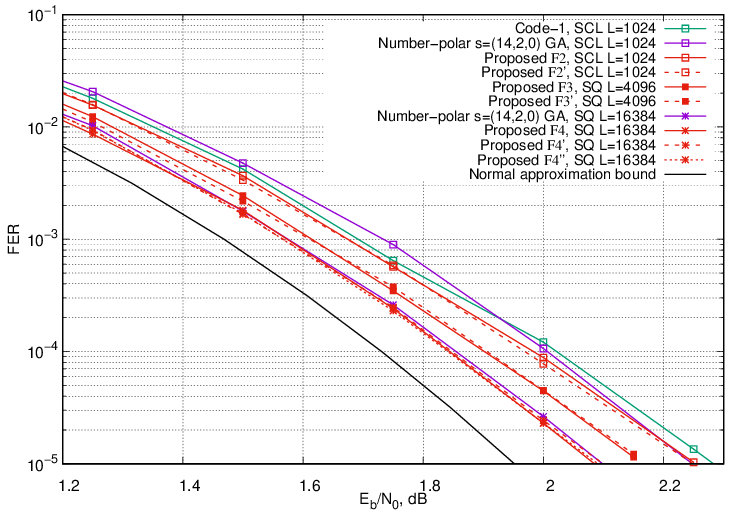} 
    \caption{{The performance comparison of  $(512,256)$ codes under SCL with $L=1024$ and SQ with $L\in\{4096,16384\}$} 
    }
 \label{fig:FER_n512_k256_fromL1024}
\end{figure}

\begin{figure}
    \centering
    \includegraphics[width=0.47\textwidth]{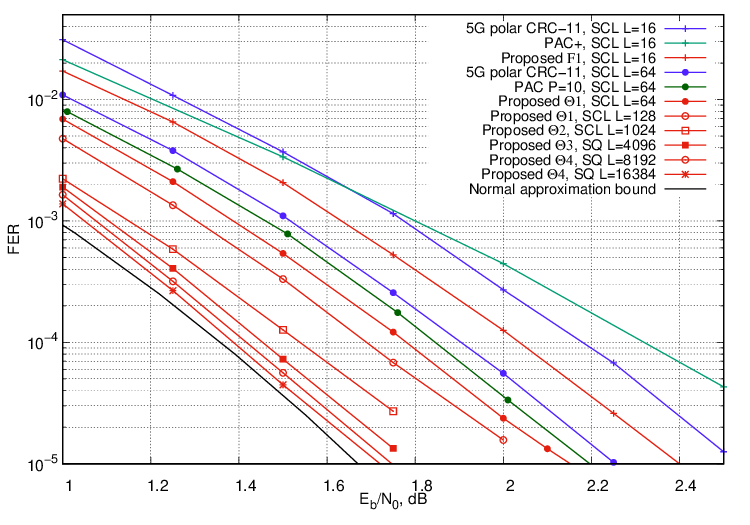} 
    \caption{{The performance comparison of $(512,128)$ codes under SCL with $L\in\{ 16,64,128,1024\}$ and SQ with $L\in\{4096,8192,16384\}$}} 
    \label{fig:FER_n512_k128}
\end{figure}

\begin{figure} 
    \centering \includegraphics[width=0.47\textwidth]{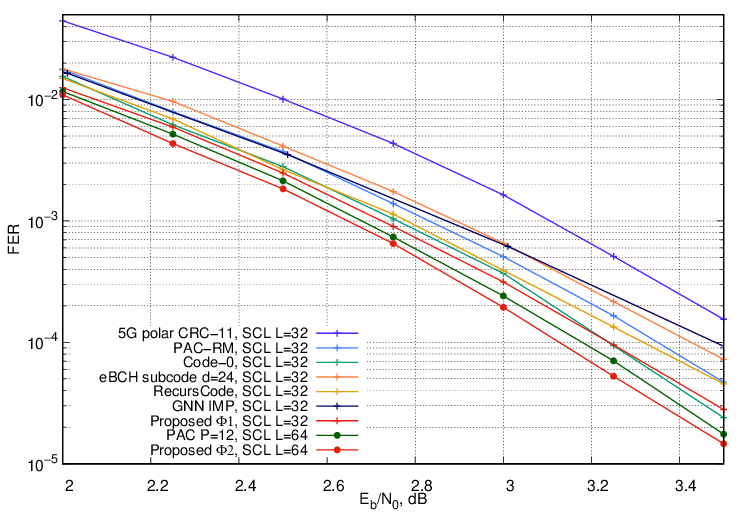} 
    \caption{{The performance comparison of  $(128,64)$ codes under SCL decoding with $L\in\{32,64\}$}}
    \label{fig:FER_n128_k64_L32_L64}
\end{figure}

\begin{figure} 
    \centering \includegraphics[width=0.47\textwidth]{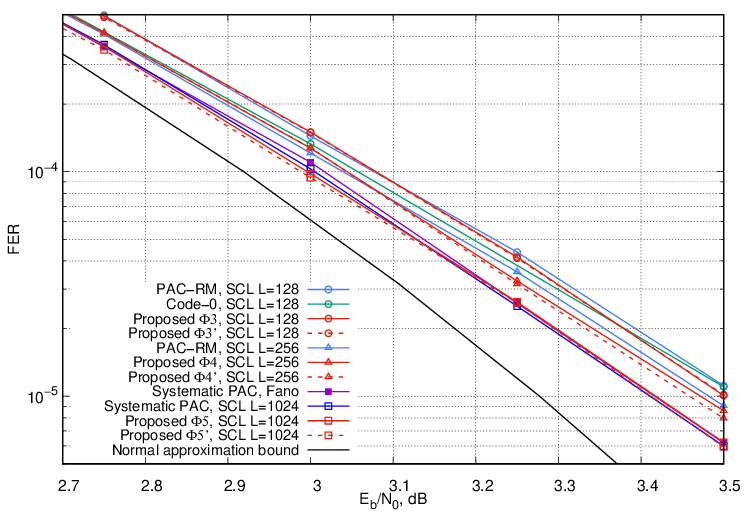} 
    \caption{{The performance comparison of  $(128,64)$ codes under SCL decoding with $L\in\{ 128,256,1024\}$}}
    \label{fig:FER_n128_k64_fromL128}
\end{figure}


{Figs. \ref{fig:FER_n512_k256_uptoL128}--\ref{fig:FER_n512_k128} show the FER performance of  precoded polar codes of length $512$ under SCL/SQ decoding. The values of list size  $L$ for SCL/SQ decoding are chosen to enable a comparison with the published results: list sizes $L\in\{32,64,128,1024,4096,16384\}$ for $(512,256)$ codes and $L\in\{16,64,128,1024,4096,8192,16384\}$ for $(512,128)$ codes.} 
It can be seen that the proposed codes outperform the state-of-the-art codes for all considered $L$. 
In {Figs. \ref{fig:FER_n128_k64_L32_L64} and \ref{fig:FER_n128_k64_fromL128}}, the proposed $(128,64)$ codes exhibit similar performance to the state-of-the-art codes under SCL {decoding} with the list sizes {$L\in\{32,64,128,256,1024\}$}. 
This is because (i) the state-of-the-art $(128,64)$ frozen sets are near-optimal due to a moderate {number of reasonable solutions of the frozen set design problem} for short code lengths, 
and (ii) in this paper, we have proposed a new low-complexity frozen set design with the optimization criteria derived for average frozen bit expressions. 
The problem of the frozen bit expression optimization for a given frozen set is left for future work. 

{It is worth noting that we provide simulation results for several 
frozen sets with similar pairs $(\bar D_{\mathrm{apx}},\widetilde{P}_\mathrm{ML})$ from Fig. \ref{fig:FrSetLabels}. These frozen sets are labelled similarly, e.g., $\Phi3$ and $\Phi3$' in Fig. \ref{fig:FrSetLabels}.
Figs. \ref{fig:FER_n512_k256_uptoL128}--\ref{fig:FER_n128_k64_fromL128} show that the corresponding proposed codes, indicated by the solid and dashed red curves, perform similarly in the case of decoding with a target $L$.
} 

Besides the excellent FER performance, the advantages of the proposed frozen set design over the main competitors \cite{Coskun2022InfTheor}, \cite{milos2024Reinf}, \cite{Arkan2019FromSD} and \cite{Tonnellier2021SystPAC} are as follows. Our proposed design method is fully specified, providing a clear frozen set design procedure. In contrast, \cite{Coskun2022InfTheor} offers
four exemplary frozen sets but lacks a general frozen set design procedure. 
We use deterministic frozen bit expressions, specified in Section \ref{sFrBitExpressions}, whereas \cite{Coskun2022InfTheor} uses randomized frozen bit expressions. 
The proposed B-constrained and S-constrained frozen set structures are more flexible than the triplet-tuned frozen sets from \cite{milos2024Reinf}. This flexibility offers additional opportunities for optimization at the expense of the increased number of evaluated frozen sets. Fast frozen set evaluation is enabled by the use of theoretical bounds instead of the decoding-based frozen set evaluation \cite{milos2024Reinf}. Thus, the computational complexity of the proposed frozen set design is low as shown in Section \ref{sGeneratedFrSets}.
Note that \cite{Arkan2019FromSD} suggested only a single $(128,64)$ PAC-RM code. Although the problem of designing PAC codes with arbitrary parameters has been solved in \cite{Tonnellier2021SystPAC} by using a genetic algorithm, 
it involves the weight distribution computation via decoding at each iteration of the genetic algorithm, leading to a large design complexity. Since the weight distribution is used as the optimization objective in \cite{Tonnellier2021SystPAC}, the corresponding codes perform well only under Fano decoding or SCL with huge $L$. In contrast, the proposed low-complexity frozen set design method is suitable for various $L$ and various code parameters.


\section{Conclusion}

In this paper, we proposed a new low-complexity frozen set design for precoded polar codes with near-uniformly distributed frozen bit expressions. The frozen set design criteria are given by analytical bounds on the FER performance and SCL complexity, where the proposed SCL complexity criterion is based on the recently published complexity analysis of SCL with near ML performance.  
These criteria define a frozen set optimization problem, whose solutions can be efficiently found by the genetic algorithm. 
To reduce the {optimization complexity}, we imposed constraints on the frozen set structure such that the number of the genetic algorithm iterations has been reduced by $5$ and $17$ times for the code parameters $(128,64)$ and $(512,256)$, respectively. The constructed
precoded polar codes of length $512$ have a superior FER performance compared to the state-of-the-art codes under SCL-based decoding with various list sizes.

\section*{Acknowledgment}
The authors would like to thank Dr. Thibaud Tonnellier for providing the frozen sets of systematic PAC codes \cite{Tonnellier2021SystPAC} and Dr. Mustafa Cemil Coşkun for clarifying the bit-channel entropy computation \cite{Coskun2022InfTheor}.


\bibliographystyle{IEEEtranTCOM}

\end{document}